%% file: 0.tex
\documentclass[final,nomarks]{dmtcs-episciences}

\received{2021-04-28}
\revised{2022-08-03, 2023-01-24}
\accepted{2023-01-27}
\publicationdetails{23}{2023}{2}{9}{7413}

\usepackage[utf8]{inputenc}
\usepackage{subfigure}
\usepackage[round]{natbib}

\usepackage{amsmath}
\usepackage{amssymb}
\usepackage{amsthm}
\usepackage{tikz-cd}
\DeclareGraphicsExtensions{.pdf}
\graphicspath{{figs/}}
\usepackage{enumerate}
\usepackage{hyperref}
\hypersetup{%
  colorlinks=true,%
  linkcolor=blue,%
  citecolor=blue
}

\author{Luigi Santocanale\thanks{Supported by the ANR project
    LAMBDACOMB ANR-21-CE48-0017}}
\title[Bijective proofs for Eulerian numbers of types $\B$ and
$\D$]{Bijective proofs for Eulerian numbers \\ of types $\B$ and
  $\D$\footnote{This is a revised and extended version of the manuscript
  \citep{2020-ALGOS} appeared in the
  proceedings of the conference ALGOS 2020}\\
  {\normalsize Dedicated to Maurice Pouzet on the occasion of his 75th birthday}
}
\affiliation{
  LIS, CNRS UMR 7020, \\ Aix-Marseille Universit\'e, France}
\email{luigi.santocanale@lis-lab.fr}
\keywords{signed permutation, weak order, threshold graph}

\input{macros} %

\begin{document}

\maketitle
% \blfootnote{This is a revised version of the manuscript
%   \cite{2020-ALGOS} appeared in the
%   proceedings of the conference ALGOS 2020 }

\input{abstract} %%

\input{intro} %% 
\input{notation} %% 
\input{sbps} %% 
\input{corollaries} %% 
\input{Stembridge} %% 
\input{computations} %%
\input{threshold} %%
\input{countingTGs} %%
\input{ack} %%

\bibliographystyle{abbrvnat}
\bibliography{biblio}

\end{document}

%% file: abstract.tex
\begin{abstract}
  Let $\sEulA{n}{k}$, $\sEulB{n}{k}$, and $\sEulD{n}{k}$ be the \En{s}
  of the types $\A$, $\B$, and $\D$, respectively---that is, the
  number of permutations of $n$ elements with $k$ descents, the number
  of \sp{s} (of $n$ elements) with $k$ type $\B$ descents, the number
  of \esp{s} (of $n$ elements) with $k$ type $\D$ descents.
  Let $S_{n}(t) = \sum_{k = 0}^{n-1} \sEulA{n}{k}t^{k}$, 
  $B_{n}(t) = \sum_{k = 0}^{n} \sEulB{n}{k}t^{k}$, and
  $D_{n}(t) = \sum_{k = 0}^{n} \sEulD{n}{k}t^{k}$. We give bijective
  proofs of the identity
  $$B_{n}(t^{2}) = (1+t)^{n+1}S_{n}(t) - 2^{n}tS_{n}(t^{2})$$
  and of \Stem's identity
  $$D_{n}(t) = B_{n}(t) - n2^{n-1}tS_{n-1}(t)\,.$$
  These bijective proofs rely on a representation of \sp{s} as
  paths. Using this representation we also establish a bijective
  correspondence between \esp{s} and pairs $(w,E)$ with $(\setn, E)$ a
  \tg and $w$ a degree ordering of $(\setn, E)$, which we use to
  obtain bijective proofs of enumerative results for \tg{s}.
\end{abstract}

%% file: intro.tex
\section{Introduction}

The \En{s} $\sEulA{n}{k}$ count the number of permutations in the
symmetric group $\Sn$ that have $k$ descent positions.  Let us recall
that, for a permutation $w = w_{1}w_{2}\cdots w_{n} \in \Sn$ (thus,
with $w_{i} \in \set{1,\ldots ,n}$ and $w_{i} \neq w_{j}$ for
$i \neq j$), a \emph{descent} of $w$ is an index (or position)
$i \in \set{1,\ldots ,n-1}$ such that $w_{i } > w_{i +1}$.

This is only one of the many interpretations
that we can give to these numbers, see e.g. \citep{Petersen2015},
yet it is intimately order-theoretic.  The set $\Sn$ can be endowed
with a lattice structure, known as the weak (Bruhat) ordering on
permutations or Permutohedron, see e.g. \citep{GuRo63,STA1}.  Descent
positions of $w \in \Sn$ are then bijection with its lower covers, so
the \En{s} $\sEulA{n}{k}$ can also be taken as counting the number of
permutations in $\Sn$ with $k$ lower covers. In particular,
$\sEulA{n}{1} = 2^{n} - n-1$ is the number of \jirr elements in
$\Sn$. A subtler order-theoretic interpretation is given in
\citep{barnard2016canonical}: since the $\Sn$ are
(join-)semidistributive as lattices, every element can be written
canonically as the join of \jirr elements \citep{FJN}; the numbers
$\sEulA{n}{k}$ count then the permutations $w \in \Sn$ that can be
written canonically as the join of $k$ \jirr elements.

The symmetric group $\Sn$ is a particular instance of a Coxeter group,
see \citep{BB2005}, since it yields a concrete realization of the
Coxeter group $\A_{n-1}$ in the family
$\A$.  
Notions of length, descent, inversion, and also a weak order, can be
defined for elements of an arbitrary finite Coxeter group
\citep{Bjo84}. We shift the focus to the families $\B$ and $\D$ of
Coxeter groups.  More precisely, this paper concerns the \En{s} in the
types $\B$ and $\D$.  The \En $\sEulB{n}{k}$ (\resp $\sEulD{n}{k}$)
counts the number of elements in the group $\Bn$ (\resp $\Dn$) with
$k$ descent positions.  Order-theoretic interpretations of these
numbers, analogous to the ones mentioned for the standard \En{s}, are
still valid.  As the abstract group $\A_{n-1}$ has a concrete
realization as the symmetric group $\Sn$, the group $\Bn$ (\resp
$\Dn$) has a realization as the hyperoctahedral group of \sp{s} (\resp
the group of \esp{s}).  Starting from these concrete representations
of Coxeter groups of type $\B$ and $\D$, we pinpoint some new
representations of \sp{s}. Relying on these representations we provide
bijective proofs of known formulas for \En{s} of the types $\B$ and
$\D$. These formulas allow us to compute the \En{s} of the types $\B$
and $\D$ from the better-known \En{s} of type $\A$.

Let $S_{n}(t)$ and $B_{n}(t)$ be the \Ep{s} of the types $\A$ and
$\B$:
\begin{align}
  \label{def:eulAB}
  S_{n}(t) &
  \eqdef \sum_{k=0}^{n-1}\eulA{n}{k}t^{k} \,,
  &
  B_{n}(t) & \eqdef \sum_{k = 0}^{n}\eulB{n}{k}t^{k}\,.
\end{align}

In \cite[\S 13, p. 215]{Petersen2015} the following polynomial
identity is stated:
\begin{align}
  \label{eq:Petersen}
  2 B_{n}(t^{2}) & = (1 + t)^{n+1}S_{n}(t) + (1 -
  t)^{n+1}S_{n}(-t)\,.
\end{align}
Considering that, for $f(t) = \sum_{k \geq 0}a_{k}t^{k}$,
\begin{align*}
  f(t) + f(-t) & = 2 \sum_{k \geq 0} a_{2k}t^{2k}\,,
\end{align*}
the polynomial identity \eqref{eq:Petersen} amounts to the following
identity among coefficients:
\begin{align}
  \eulB{n}{k} & = \sum_{i = 0}^{2k}\eulA{n}{i}\binom{n+1}{2k-i}\,.
  \label{eq:eulBeven}
\end{align}
We present a bijective proof of \eqref{eq:eulBeven} and also establish
the identity
\begin{align}
  2^{n}\eulA{n}{k} & = \sum_{i = 0}^{2k+1}\eulA{n}{i}\binom{n+1}{2k
  +1-i}\,.
  \label{eq:eulBodd}
\end{align}
Considering that, for $f(t) = \sum_{k \geq 0}a_{k}t^{k}$,
\begin{align*}
  f(t) - f(-t) & = 2 \sum_{k \geq 0} a_{2k +1}t^{2k +1}\,,
\end{align*}
the identity \eqref{eq:eulBodd} yields the polynomial identity:
\begin{align*}
  2^{n+1}tS_{n}(t^{2})
  & = 
  (1 + t)^{n+1}S_{n}(t) - (1-t)^{n+1}S_{n}(-t) 
  \,.
\end{align*}
More importantly, \eqref{eq:eulBeven} and \eqref{eq:eulBodd} jointly
yield the polynomial identity
\begin{align}
  (1 + t)^{n+1}S_{n}(t) & = B_{n}(t^{2}) + 2^{n}tS_{n}(t^{2})\,.
  \label{eq:main}
\end{align}
Let now $D_{n}(t)$ be the \Ep of type $\D$:
\begin{align*}
  D_{n}(t) & \eqdef \sum_{k = 0}^{n} \eulD{n}{k} t^{k}\,.
\end{align*}
Investigating further the terms $2^{n}S_{n}(t)$, we could 
find a simple bijective proof, that we present here, of \Stem's
identity \cite[Lemma 9.1]{Stembridge1994} 
\begin{align}
  \label{eq:Stembridge}
  D_{n}(t) & = B_{n}(t) - n 2^{n-1}tS_{n-1}(t)\,,
\end{align}
which, in terms of the \En{s} of type $\D$, amounts to
\begin{align*}
  \eulD{n}{k} & =  \eulB{n}{k} -n2^{n-1}\eul{n-1}{k-1}\,.
\end{align*}
The proofs presented here differ from known proofs of the identities
\eqref{eq:Petersen} and \eqref{eq:Stembridge}. As suggested in
\citep{Petersen2015}, the first identity may be derived by computing
the $f$-vector of the type $\B$ Coxeter complex and then by applying
the transform yielding the $h$-vector from the $f$-vector. A similar
method is used in \citep{Stembridge1994} to prove the identity
\eqref{eq:Stembridge}.  
Our proofs directly rely on the combinatorial properties of
signed/\esp{s} and on two representations of \sp{s} that we call one
the \emph{\prep} and, the other, the \emph{\sbp} representation.
The idea is that a \sp of $\setn$ can be organised into a discrete
path from $(0,n)$ to $(n,0)$ that only uses East and \Ss{s} and that,
by projecting onto the $x$-axis, we obtain a permutation divided into
blocks, as suggested in Figure~\ref{fig:TwoRepresentations}.
\begin{figure}
  \centering
  \includegraphics{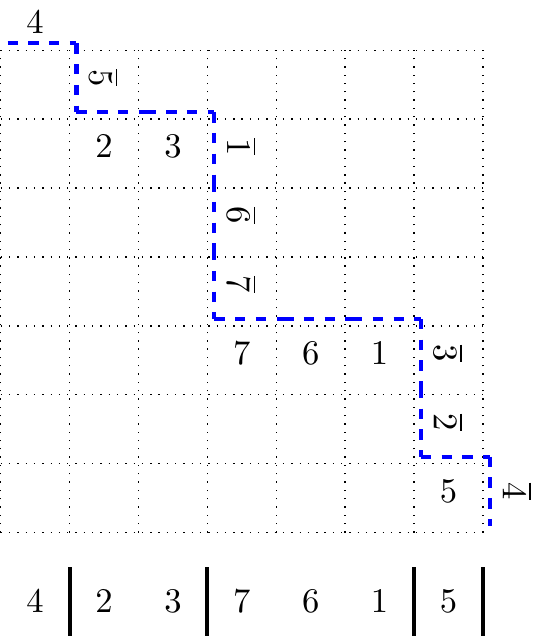}
  \caption{\Sp{s} as paths and as barred permutations}
  \label{fig:TwoRepresentations}
\end{figure}
As a byproduct of these representations, we also obtain a bijection
between \esp{s} of $\setn$ and pairs $(w,E)$ where $(\setn,E)$ is a
\tg and $w$ is a permutation or, better, a linear ordering of $\setn$
that is a \dgo for $(\setn,E)$. Under the bijection, the ordering of
$\Dn$ is coordinatewise, that is, we have
$(w_{1},E_{1}) \leq (w_{2},E_{2})$ if and only if $w_{1} \leq w_{2}$
in $\Sn$ and $E_{1} \subseteq E_{2}$.  It is in the scope of future
research to shed some light on the lattice structure of the weak order
on the Coxeter groups of type $\D$ using this representation of the
ordering. Our confidence that this is indeed possible relies on our
progress studying these lattices, which yielded the discovery of the
path representation of \sp{s}.
In the meantime, the two representations of \sp{s}, together with
  the bijection between \esp{s} and pairs $(w,E)$ as mentioned above,
  also yield a representation of \tg{s} as specific \sbp{s}.  We
  exemplify once more the convenience of these representations by
  deducing enumerative results for \tg{s} \citep{BP1987,Spiro2020}.

%% file: notation.tex
\section{Background}
\label{sec:notation}
The notation used is chosen to be consistent with \citep{Petersen2015}.
We use $\setn$ for the set $\set{1,\ldots ,n}$
and $\Sn$ for the set of permutations of $\setn$.  We use $\setnp$ for
the set $\set{0,1,\ldots ,n}$, $\setmn$ for $\set{-n,\ldots ,-1}$, and $\setpmn$ for
$\set{-n,\ldots ,-1,1,\ldots ,n}$.
We write a permutation $w \in \Sn$ as a word
$w = w_{1}w_{2}\cdots w_{n}$, with $w_{i} \in \setn$.  For
$w \in \Sn$, its set of descents and its set of inversions\footnote{
  It is also possible to define $\inv{w}$ as the set
  $\set{(i,j) \mid 1 \leq i < j \leq n, w_{i} > w_{j}}$.  The
  definition given above is better suited for the order-theoretic
  approach.
} are defined as follows:
\begin{align}
  \nonumber
  \desc{w}& \eqdef \set{i \in \set{1,\ldots ,n-1}\mid w_{i} > w_{i
      +1}}\,,
   \\
  \inv{w} & \eqdef \set{(i,j) \mid 1 \leq i < j \leq n, w^{-1}(i)
    > w^{-1}(j)}\,.
  \intertext{Then, we let}
  \des{w} & \eqdef \Card{\desc{w}}\,.
\end{align}
Notice that a descent $i$ of a permutation $w_{1}w_{2}\cdots w_{n}$ is
uniquely identified by the two contiguous letters $w_{i}w_{i+1}$ such
that $w_{i} > w_{i+1}$. Therefore, we shall often identify such a
descent by these two contiguous letters. 
The Eulerian number $\sEulA{n}{k}$, counting the number of
permutations of $n$ elements with $k$ descents, can be formally
defined as follows:
\begin{align*}
  \eulA{n}{k} & \eqdef \Card{\set{w \in \Sn \mid \des{w} = k}}\,.
\end{align*}
Let us define a \emph{\sp of $\setn$} as a permutation $u$ of
$\setpmn$ such that, for each $i \in \setpmn$, $u_{-i} = -u_{i}$.  We
use $\Bn$ for the set of \sp{s} of $\setn$. When writing a \sp $u$ as
a word $u_{-n}\cdots u_{-1}u_{1}\cdots u_{n}$, we prefer writing
$u_{i} = \wbar{x}$ in place of $-x$ if $u_{i} < 0$ and
$\abs{u_{i}} = x$. 
Also, we often write $u \in \Bn$ in \emph{window notation}, that is,
we only write the suffix $u_{1}u_{2}\cdots
u_{n}$; 
indeed, the prefix $u_{-n}u_{n-1}\cdots u_{-1}$ is determined as the
mirror of the suffix $u_{1}u_{2}\cdots u_{n}$ up to exchanging the
signs. The set $\Bn$ is a subgroup of the group of permutations of
the set $\setpmn$ and, as mentioned before, it is the standard model
for the Coxeter group in the family $\B$ with $n$ generators.
General notions from the theory of Coxeter groups (descent, inversion)
apply to \sp{s}.

The Cayley graph of a Coxeter group (which, by definition, comes with
a set of generators) can always be oriented by increasing length,
where the length of an element is defined as the number of its
inversions.  The oriented graph is then the Hasse diagram of a poset
where descents of an element mark its lower
covers.  
While for permutations the notions of descent, inversion, and length
are customary from elementary combinatorics, for \sp{s} these
notions are subtler yet well studied, we refer to standard
monographs such as \citep{BB2005,Petersen2015}. 
We present below, as definitions, the well-known
explicit formulas for the descent and inversion sets of $u \in
\Bn$. We let
\begin{align}
  \label{def:descinvB}
  \nonumber
  \descB{u}& \eqdef \set{i \in \set{0,\ldots ,n-1}\mid u_{i} > u_{i
      +1}}\,, 
  \\   \invB{u} & \eqdef \set{(i,j) \mid 1 \leq \abs{i} \leq j \leq n, u^{-1}(i)
      > u^{-1}(j)}\,,
  \intertext{where we set $u_{0} \eqdef 0$, so $0$ is a descent of $u$
    if and only if $u_{1} < 0$,} 
  \nonumber
  \desB{u} & = \Card{\descB{u}}\,,
  \\\eulB{n}{k} & \eqdef \Card{\set{u \in \Bn \mid \desB{u} =
    k}}\,.
\end{align}
The definition of the \Ep{s} in the types $\A$ and $\B$ appears in
\eqref{def:eulAB}. Let us mention that the type $\A$ \Ep is often (for
example in \citep{Bona2012}) defined as follows:
\begin{align*}
  A_{n}(t) & \eqdef \sum_{k = 1}^{n}\eulA{n}{k-1}t^{k}
  = t S_{n}(t)\,.
\end{align*}
We exclusively manipulate the polynomials $S_{n}(t)$ and never the
$A_{n}(t)$.  Notice that $S_{n}(t)$ has degree $n-1$ and $B_{n}(t)$
has degree $n$.

We shall introduce later \esp{s} and their groups, as well as related
notions arising from the fact that these groups are standard models
for Coxeter groups in the family $\D$.

\medskip

For $u \in \Bn$ we let
$\descBp{u} \eqdef \descB{u} \setminus \set{0}$, that is, $\descBp{u}$
is the set of strictly positive descents of $u$.
Let us observe the following:
\begin{lemma}
  \label{lemma:positiveDescs}
  $\Card{\,\set{u \in \Bn \mid \Card{\descBp{u}} = k}\,} = 2^{n}\sEulA{n}{k}$.
\end{lemma}
\begin{proof}
  Recall that the window notation of a \sp $u$ is the word
  $u_{1}\cdots u_{n}$. We identify the window notation of $u$ with the
  mapping $\tilde{u} : \setn \rto \setpmn$ such that
  $\tilde{u}(i) = u_{i}$, for each $i \in \setn$.

  We claim that maps arising as window notation of a \sp are in
  bijection with pairs $(w,\iota)$ where $w \in \Sn$ and
  $\iota : \setn \rto \setpmn$ is an order preserving injection such
  that $x \in \iota(\setn)$ iff $-x \not \in \iota(\setn)$.
  The bijection goes as follows.  Given a \sp $u$, the image
  $\tilde{u}(\setn) \subseteq \setpmn$ of its window notation is a
  subset of integers of cardinality $n$ with the linear ordering
  inherited from integers. Thus, there exists a unique order
  preserving bijection $\psi : \tilde{u}(\setn) \rto \setn$. We
  associate to $\tilde{u}$ the pair $(w,\iota)$ where
  $w = \psi \circ \tilde{u}$ and where $\iota$ is the composition of
  $\psi^{-1}: \setn \rto \tilde{u}(\setn)$ and the inclusion
  $\tilde{u}(\setn) \subseteq \setpmn$. Notice that
  $\tilde{u} = \iota \circ w$, from which it follows that
  $\tilde{u}(\setn) = \iota(\setn)$ and that $x \in \iota(\setn)$ iff
  $-x \not \in \iota(\setn)$. Let us argue that this decomposition is
  unique.  Notice that $\tilde{u}(\setn) = \iota(\setn)$ uniquely
  determining $\iota$. Since moreover $\iota $ is injective, if
  $\iota \circ w = \tilde{u} = \iota \circ w'$, then $w = w'$ as well.
  Also, given such a pair $(w,\iota)$, $\iota \circ w$ is a preimage
  of $(w,\iota)$ via the correspondence, which is therefore surjective
  and a bijection.

  Next, the order preserving injections $\iota : \setn \rto \setpmn$
  such that $x \in \iota(\setn)$ iff $-x \not \in \iota(\setn)$ are
  uniquely determined by their positive image
  $\iota(\setn) \cap \setn$, so there are $2^{n}$ such mappings.
  Finally, let us argue that, for $i = 1,\ldots ,n-1$,
  $u_{i} > u_{i + 1}$ if and only if $w_{i} > w_{i + 1}$:
  $u_{i} > u_{i + 1}$ iff $\tilde{u}(i) > \tilde{u}(i + 1)$ iff
  $\iota(w(i)) > \iota(w(i + 1))$ iff $w_{i} > w_{i +1}$, since
  $\iota$ preserves and reflects the ordering. It follows that
  $\descBp{u} = \desc{w}$, thus yielding the statement of the lemma.
\end{proof}

\begin{example}
  Consider the signed permutation
  $u \eqdef 3\wbar{4}1\wbar{2}\wbar{5}$.  Then
  $\tilde{u} = \iota \circ w$ with $w = 52431$ and $\iota$ the order
  preserving map $\bar{5}\bar{4}\bar{2}13$ with
  $\iota(\setn) \cap \setn = \set{1,3}$.  \EndOfExample
\end{example}

Let us end this section with a notational remark.
We can index cells
or points of a grid such as the one in
Figure~\ref{fig:TwoRepresentations} in two different ways. Either we
consider them as being part of a matrix and index them by row (we
count here rows from the bottom to the top) and column, thus using the
letters $i,j$.
Or we can index them using the abscissa and ordinate of the
two-dimensional plane.  We shall prefer the latter method when the
axes are ordered in the standard way, and the first method when the
axes are ordered according to a permutation. For example, in
Figure~\ref{fig:TwoRepresentations}, the dashed path makes an
East-South turm at point $(x,y)$ with $x = 3$ and $y = 6$. However, if
we relabel the axes by means of the permutation $w = 4237615$, then it
makes sense to say that the path makes an East-South turn at row
$i = w(6) = 1$ and column $j = w(3) = 3$.  The use of both kind of
indexing shall be unavoidable. The reader should be aware that rows
correspond to the ordinate and columns to the abscissa, so a point
$(x,y)$ yields the cell $M_{w(y),w(x)}$ of a matrix $M$ and that a
cell $M_{i,j}$ is located at point $(w^{-1}(j),w^{-1}(i))$.

%% file: sbps.tex
\section{\Prep of \sp{s}} 

We present here our main combinatorial tool to deal with \sp{s}, the \prep.

\begin{definition}
  The \emph{\prep} of $u \in \Bn$ is a triple
  $(\pi^{u},\lux,\luy)$ where $\pi^{u}$ is a
  discrete path, drawn on a grid $\setnp \times \setnp$ and joining
  the point $(0,n)$ to the point $(n,0)$,
  $\lux : \setn \rto \setn$, and
  $\luy : \setn \rto \setmn$. The triple
  $(\pi^{u},\lux,\luy)$ is constructed from $u$
  according to the following algorithm:
  \begin{inparaenum}[(i)]
  \item $u$ is written in full notation as a word and scanned from
    left to right: each
    positive letter yields an \Es (a length $1$ step along the
    $x$-axis towards the right), and each negative letter yields a \Ss
    (a length $1$ step along the $y$-axis towards the bottom) ;
  \item the labelling $\lux : \setn \rto \setn$ is obtained by
    projecting each positive letter on the $x$-axis,
  \item the labelling $\luy : \setn \rto \setmn$ is obtained by
    projecting each negative letter on the $y$-axis.
  \end{inparaenum}
\end{definition}

\begin{example}
  \label{example:prep}
  Consider the \sp $u \eqdef \wbar{2}316\wbar{4}\wbar{7}5$, in window
  notation, that is,
  $\wbar{5}74\wbar{613}2\wbar{2}316\wbar{4}\wbar{7}5$, in full
  notation.  Applying the algorithm above, we draw the path $\pi^{u}$
  and the labellings $\lux,\luy$ as follows:
  \begin{center}
    \includegraphics[page=1]{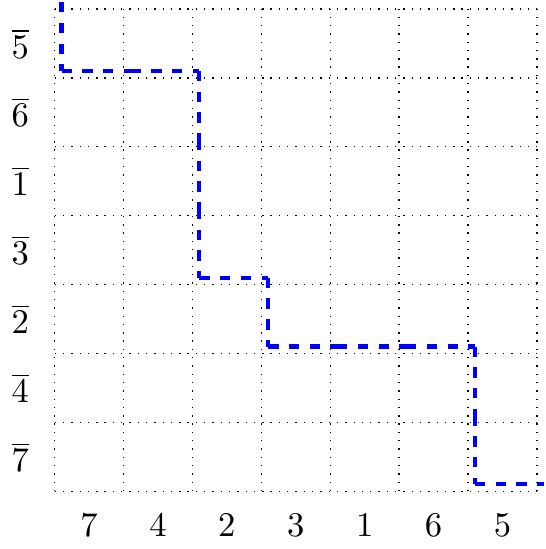}
  \end{center}
  Therefore, $\pi^{u}$ is the dashed 
  path, $\lux$ is
  the permutation $7423165$, and $\luy$ is
  $\wbar{7}\,\wbar{4}\,\wbar{2}\,\wbar{3}\,\wbar{1}\,\wbar{6}\,\wbar{5}$.
  \EndOfExample
\end{example}
It is easily seen that, for an arbitrary $u \in \Bn$,
$(\pi^{u},\lux,\luy)$ has the following
properties: 
\begin{enumerate}[(i)]
\item $\pi^{u}$ is symmetric along the diagonal,
\item $\lux \in \Sn$ and, moreover, it is the subword of
  $u$ of positive letters,
\item for each $x \in \setn$, $\luy(x) = \wbar{\lux(x)}$
  and, moreover, $\luy$ is the mirror of the subword of $u$
  of negative letters.
\end{enumerate}
In particular, we see that the data
$(\pi^{u},\lux,\luy)$ is redundant since
$\luy$ is completely determined from $\lux$.
\begin{proposition}
  The mapping $u \mapsto (\pi^{u},\lux)$ is a bijection
  from the set of \sp{s} $\Bn$ to the set of pairs $(\pi,w)$, where
  $w \in \Sn$ and $\pi$ is a discrete path from $(0,n)$ to $(n,0)$
  with \ESs{s} which, moreover, is symmetric along the diagonal.
\end{proposition}
We leave it to the reader to verify the above statement.
Next, we argue for the interest of this representation by looking at
the inversion set of a \sp.  According to the definition in
\eqref{def:descinvB}, the type $\B$ inversions of a \sp can be split
into its positive inversions, the pairs $(i,j)$ with
$1 \leq i < j \leq n$, and the negative ones, those of the form
$(i,j)$ with $i \leq 0$ and $1 \leq \abs{i} \leq j \leq n$.
We claim that the positive inversions of $u$ are the type $\A$
inversions of $\lux$ and that its negative inversions are of the form
$(\luy(y),\lux(x))$ such that $1\leq y \leq x \leq n$ and the cell
$(x,y)$ lies below $\pi^{u}$.
This idea is exemplified in Figure~\ref{fig:negInvs} with the \sp
$\wbar{2}316\wbar{4}\wbar{7}5$ from Example~\ref{example:prep}.
\begin{figure}[h!]
  \centering
  \includegraphics[page=2]{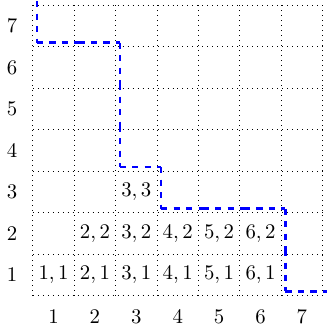}
  \qquad\qquad
  \includegraphics[page=3]{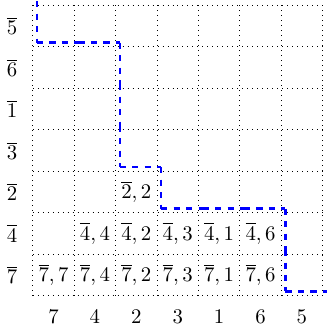}
  \caption{Negative inversions of $\wbar{2}316\wbar{4}\wbar{7}5$,
    indexed on the left by the abscissa and ordinate and on the
    right by $\luy,\lux$}
  \label{fig:negInvs}
\end{figure}
Notice that, when $1 \leq x < y \leq n$ and $(x,y)$ lies below
$\pi^{u}$, then $(y,x)$ lies below $\pi^{u}$ as well (since $\pi^{u}$
is symmetric along the diagonal) and therefore, according to our
claim, $(\luy(x),\lux(y))$ is a negative inversion of $u$.
If we identify, when $x < y$, the pair $(\luy(y),\lux(x))$ with
$(\luy(x),\lux(y))$, then we can simply say that the negative
inversions of $u$ are of the form $(\luy(y),\lux(x))$ for $(x,y)$
below $\pi^{u}$.

We collect these observations in a formal statement.
\begin{proposition}
  \label{prop:charInversions}
  Let $u \in \Bn$.  For each $i,j$ with
  $1 \leq \abs{i} \leq j \leq n$, $(i,j) \in \invB{u}$ if and only if
  either $1 \leq i < j \leq n$ and $(i,j) \in \inv{\lux}$
  or $i < 0$ and
  $((\lux)^{-1}(\minus{i}),(\lux)^{-1}(j))$ lies 
  below the path $\pi^{u}$.
\end{proposition}
\begin{proof}
  Consider a pair $(i,j)$ such that $1 \leq \abs{i} \leq j \leq n$ and
  such that, if $0 < i$, then $i < j$.

  If $0 < i < j$, then both $i$ and $j$ 
  appear in $\lux$, which is the subword of $u$
  (written in full notation) of positive integers. Then
  $u^{-1}(i) > u^{-1}(j)$ if and only if
  $(\lux)^{-1}(i) > (\lux)^{-1}(j)$, that is,
  $(i,j) \in \invB{u}$ if and only if
  $(i,j) \in \inv{\lux}$.

  We suppose next that $i < 0$.  Observe that, as suggested in
  Figure~\ref{fig:xyWrtBluePath}, for $j > 0$ and $i < 0$, the cell
  identified by $\luy,\lux$ as $(i,j)$ is below $\pi^{u}$ if and only
  if the letter $j$ appears before the letter $i$ in $u$. Also, for
  such a pair, $j$ appears before $i$ in $u$ if and only if
  $(i,j) \in \inv{u}$, where $u$ is considered as a permutation of the
  set $\setpmn$ and the set of inversions is computed w.r.t the
  standard linear order on this
  set. 
  
  Therefore, if $(i,j)$ with $i < 0$ and
  $1 \leq \abs{i} \leq j \leq n$, then $(i,j) \in \invB{u}$ if and
  only if $(i,j) \in \inv{u}$ if and only if the cell identified by
  $\luy,\lux$ as $(i,j)$ is below $\pi^{u}$. 
  If, instead of using $\luy$ and $\lux$ to identify cells, we use the
  abscissa and ordinate, this happens when
  $((\luy)^{-1}(i),(\lux)^{-1}(j)) =
  ((\lux)^{-1}(\minus{i}),(\luy)^{-1}(j))$ is below $\pi_{u}$.
\end{proof}

\begin{figure}[h!]
  \centering
  \includegraphics[scale=0.8]{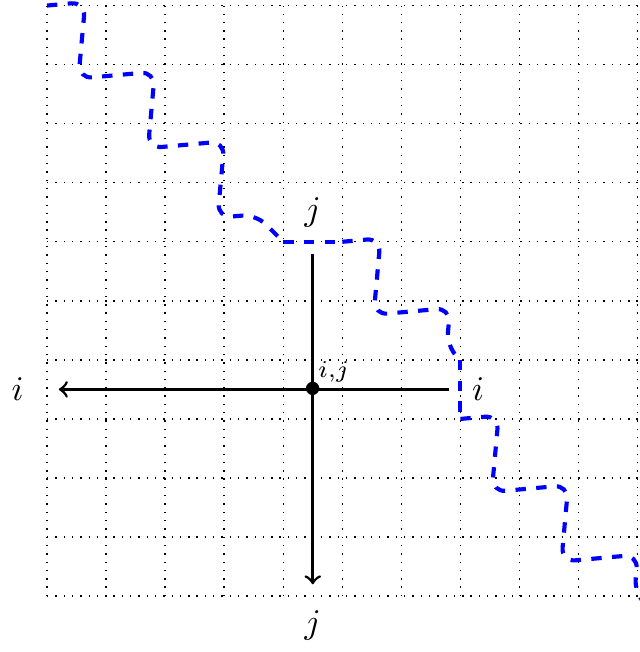}
  \caption{Characterizing inversions of the form $(i,j)$ with $i$
    negative}
  \label{fig:xyWrtBluePath}
\end{figure}

\begin{remark}
  \label{rem:charInversions}
  The fact that $(x,y)$ lies below $\pi^{u}$ if and only if $(y,x)$
  lies below $\pi^{u}$ suggests to look at negative inversions of $u$
  as unordered pairs of the form $\upair{\lux(x),\lux(y)}$ (doubletons
  or singletons) such that $(x,y)$ lies below $\pi^{u}$.  We shall
  explore this graph-theoretic approach in Section~\ref{threesholds}.
  We illustrate this with the \sp 
  of Example~\ref{example:prep}: we can identify the set of type $\B$
  inversions of $\wbar{2}316\wbar{4}\wbar{7}5$ with the disjoint union
  of the set of type $\A$ inversions of $7423165$ and the set of
  unordered pairs 
  \begin{align*}
    \set{\upair{7,7}, \upair{7,4}, & \upair{7,2}, \upair{7,3},
      \upair{7,1}, \upair{7,6}, 
      \upair{4,4}, \upair{4,2}, \upair{4,3}, \upair{4,1}, \upair{4,6},
      \upair{2,2}}\,.
  \end{align*}
\end{remark}

\section{Simply barred permutations}
\label{sec:sbps}
We consider now a second way of representing \sp{s}.  We mostly
consider \sbp{s} as shorthands for \prep{s} of \sp{s}. While less
informative than \prep{s}, we shall observe that the enumerative
results of the following chapters mostly rely on this representation.
\begin{definition}
  A \emph{\sbp} of $\setn$ is a pair $(w,B)$ where $w \in \Sn$ and
  $B \subseteq \set{1,\ldots ,n}$. We let $\SBPn$ be the set of
  \emph{\sbp}{s} of $\setn$.
\end{definition}
We think of $B$ as a set of positions of $w$, the barred positions or
walls. 
We have added the adjective ``simply'' to ``barred permutation'' since
we do not require that $B$ is a superset of $\Desc{w}$, as for example
in \citep{GS78}. 
\begin{ex}
  We write a \sbp $(w,B)$ as a permutation divided into \emph{blocks}
  by the bars, placing a vertical bar after $w_{i}$ for each
  $i \in B$.  For example, $(w,B)= (7423165, \set{2,4,6})$ is written
  $74\vbar 23 \vbar 16 \vbar 5$.  
  Notice that we allow a bar to appear in the last position, for
  example $34\vbar1\vbar 265\vbar7\vbar$ stands for the \sbp
  $(3412657, \set{2,3,6,7})$. Thus, a bar appears in the last position
  if and only if the last block is empty.  The last block is indeed
  the only block that can be empty, which amounts to saying that
  consecutive bars are not allowed in \sbp{s}. This contrasts with
  other notions of barred permutation, for example the one appearing
  in the proof of the alternating sum formula for the \En{s}
  \citep[Theorem 1.11]{Bona2012}.
\end{ex}

Next, we describe a bijection---that we call $\psi$---from the set
$\SBPn$ to $\Bn$.  Let us notice that, in order to establish
equipotence of these two sets, other more straightforward bijections
are available. In the definition below, for
$w = w_{1}w_{2}\cdots w_{n}$, we let
$\wbar{w} \eqdef \wbar{w_{1}}\wbar{w_{2}}\cdots
\wbar{w_{n}}$. 
Moreover, we need to explain what we mean for the \emph{upper
  \antidiagonal} of a subgrid.  Notice that, in a grid
$\setnp \times \setnp$, we have two discrete paths that are closest to
the the line $y = n -x$. We call them the lower and upper (discrete)
\antidiagonal, respectively. This is suggested below with the lower
and the upper \antidiagonal on the left and on the center,
respectively.
\begin{center}
  \includegraphics[scale=0.9]{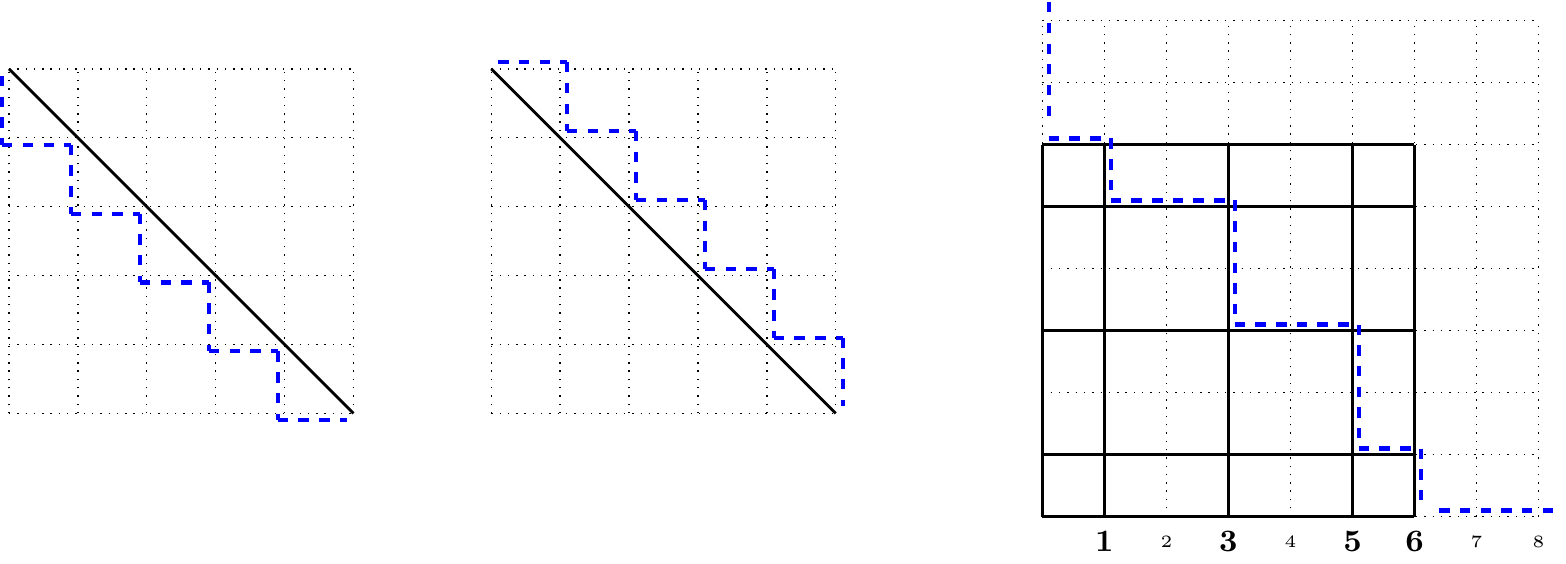}
\end{center}
Next, a subset $B \subseteq \setn$ determines a subgrid
$(B \cup \set{0}) \times (B \cup \set{0})$ which comes with its own
upper antidiagonal. We can extend this path with South steps before
and East steps after, so to obtain a path from $(0,n)$ to
$(n,0)$. We call this path the \emph{upper antidiagonal of the
  subgrid}. An example appears above on the right, where the subgrid
is determined by $B = \set{1,3,5,6}$.

\begin{definition}
  \label{def:psi}
  For $(w,B) \in \SBPn$, we define the \sp $\psi(w,B) \in \Bn$
  according to the following algorithm:
  \begin{inparaenum}[(i)]
  \item 
    draw the grid $\setnp \times \setnp$; 
  \item since $B \subseteq \setn$,
    $(B \cup \set{0})\times (B\cup \set{0})$ defines a subgrid of
    $\setnp \times \setnp$, construct the upper \antidiagonal $\pi$ of
    this subgrid;
  \item $\psi(w,B)$ is the \sp $u$ whose \prep
    $(\pi^{u},\lux,\luy)$ equals to $(\pi,w,\wbar{w})$.
  \end{inparaenum}
\end{definition}

\begin{ex}
  \label{example:grid}
  The construction just described can be understood as raising the
  bars and transforming them into a grid.  For example, for the \sbp
  $74\vbar 2\vbar 31 6\vbar 5$ (that is, $(w,B)$ with $w = 7423165$ and
  $B = \set{2,3,6}$) the construction is as follows:
  \begin{center}
    \includegraphics{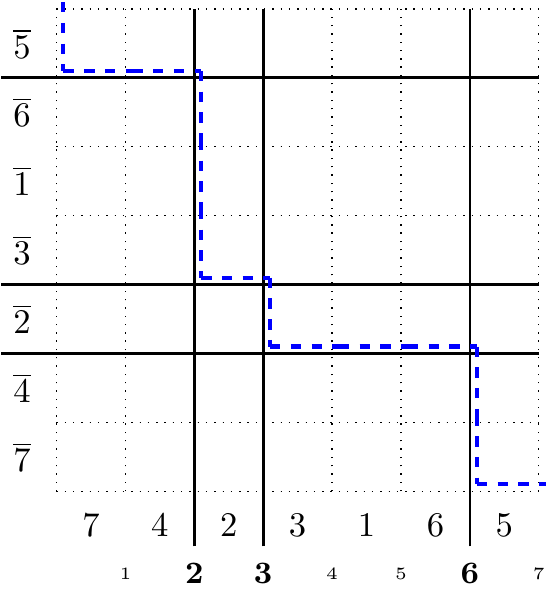}
  \end{center}
  The upper \antidiagonal of the subgrid 
  yields the
  dashed path above. The resulting \sp $\psi(w,B)$ is
  $\wbar{2}316\wbar{4}\wbar{7}5$ as from
  Example~\ref{example:prep}.
\end{ex}

The inverse image of $\psi$ can be constructed according to the
following algorithm: for $u \in \Bn$
\begin{inparaenum}[(i)]
\item construct the  \prep $(\pi^{u},\lux,\luy)$
  of $u$,
\item insert a bar in $w$ at each vertical step of $\pi^{u}$ (and remove
  consecutive bars),
\item remove a bar at position $0$ if it exists.
\end{inparaenum}
Said otherwise, $(w,B) = \psi^{-1}(u)$ is obtained from $u$ by
transforming each negative letter into a bar, by removing consecutive
bars, and then by removing a bar at position $0$ if needed.

Even if we consider \sbp{s} as shorthands for \prep{s} of \sp{s}, some
remarks are due now:
\begin{lemma}
  \label{lemma:barsxyturns}
  If $u = \psi(w,B)$, then there is a bijection between the set $B$ of
  bars and the set of \xyturn{s} of $\pi^{u}$.
\end{lemma}

\begin{lemma}
  \label{lemma:ZeroDescent}
  We have $0 \in \descB{\psi(w,B)}$ if and only if $\Card{B}$ is odd.
\end{lemma}
The lemma can immediately be verified by considering that
$0 \in \descB{u}$ if and only if the first letter in the window
notation of $u$ is negative, if and only if, in the \prep of
$\psi(w,B)$, the first step of $\pi^{u}$ after meeting the diagonal is
along the $y$-axis, in which case (and only in this case) $\pi^{u}$
makes an \xyturn on the diagonal. This happens exactly when $\pi^{u}$
has an odd number of \xyturn{s}.

\section{Descents from \sbp{s}}

We start investigating how the type $\B$ descent set can be recovered
from a \sbp.
\begin{proposition}
  \label{prop:descentsFromSBPs}
  For a \sbp $(w,B)$, we have
  \begin{align}
    \label{eq:descents}
    \desB{\psi(w,B)} & = 
    \Card{\desc{w} \setminus B} +
    \ceiling{\frac{\Card{B}}{2}}\,.
  \end{align}
\end{proposition}
\begin{proof}
  Write $u = \psi(w,B)$ in window notation and divide it in maximal
  blocks of consecutive letters having the same sign. If the first
  block has negative sign, add an empty positive block in position $0$.
  Each change of sign $+-$ among consecutive blocks yields a descent.
  These changes of sign bijectively correspond to \xyturn{s} of $\pi^{u}$
  that lie on or below the diagonal.  By
  Lemma~\ref{lemma:barsxyturns}, each bar determines an \xyturn and,
  by symmetry of $\pi^{u}$ along the diagonal, the number of
  \xyturn{s} that are on or below the diagonal is
  $\ceiling{\frac{\Card{B}}{2}}$. Therefore this quantity counts the
  number of descents determined by a change of sign.

  The other descents of $\psi(w,B)$ are either of the form
  $w_{i}w_{i+1}$ with $w_{i} > w_{i+1}$ and $w_{i},w_{i+1}$ belonging
  to the same positive block, or of the form
  $\wbar{w_{i+1}}\wbar{w_{i}}$ with $w_{i} > w_{i+1}$ and
  $\wbar{w_{i}},\wbar{w_{i+1}}$ belonging to the same negative block.
  These descents are in bijection with the descent positions of $w$
  that do not belong to the set $B$.
\end{proof}

For each $k \in \set{0,1,\ldots ,n}$, we let in the following $\SBPnk$
be the set of \sbp{s} $(w,B) \in \SBPn$ such that
$\Card{\desc{w} \setminus B} + \myexpr{B} = k$.
\begin{corollary}
  The set $\SBPnk$ is in bijection with the set of \sp{s} of $\setn$
  with $k$ descents.
\end{corollary}

We introduce next \wbp{s} only as a tool to index \sbp{s}
independently of the even/odd cardinalities of their set of bars.
\begin{definition}
  A \emph{\wbp} of $\setn$ is a pair $(w,B)$ where $w$ is a permutation
  of $\setn$ and $B \subseteq \set{0,\ldots ,n}$ is a set of positions
  (the bars). We let $\WBPn$ be the set of
  \emph{\wbp}{s} of $\setn$.
\end{definition}

For $D \subseteq \setn$, let $\xiD : P(\setnp) \rto P(\setn)$ be the
map defined by
\begin{align*}
  \xiD(B) & \eqdef (D \Delta B) \setminus \set{0} = D \Delta (B
  \setminus \set{0}) \,,
\end{align*}
where $\Delta$ stands for the symmetric difference in $P(\setnp)$.
Then, we define $\Theta_{n} : \WBPn \rto \SBPn$ by
\begin{align*}
  \Theta_{n}(w,B) & \eqdef  (w,\xiD[\desc{w}](
  B))\,.
\end{align*}
We shall investigate properties of the map $\Theta_{n}$, for which
we first need to collect properties of the map $\xiD$.
These are listed in the following lemmas.
\begin{lemma}
  The map $\xiD$ is a surjective two-to-one map. That is, each
  $C \subseteq \setn$ has exactly two preimages,
  $B_{1} \eqdef D \Delta C$ and $B_{2} = (D \Delta C) \cup \set{0}$.
\end{lemma}
In view of $D \setminus \xiD(B) = D \cap B$, the following relation
holds if $0 \not\in B$:
\begin{align*}
  \Card{D} + \Card{B} & = 2\Card{D \setminus \xiD(B)} + \Card{\xiD(B)}\,.
\end{align*}
Given this relation, the reader shall have no difficulties verifying
the properties stated in the next lemma.
\begin{lemma}
  \label{lemma:counting}
  The parity of $\Card{\xiD(B)}$ can be computed from
  $\Card{D}$ and $\Card{B}$ according to the following rules:
  \begin{enumerate}
  \item if $\Card{D} + \Card{B} = 2k$ and $0 \not\in B$, then
    $\Card{D \setminus \xiD(B)} + \MyExpr{B} = k$ and, in this
    case, $\Card{\xiD(B)}$ is even;
  \item if $\Card{D} + \Card{B} = 2k$ and $0 \in B$, then
    $\Card{D \setminus \xiD(B)} + \MyExpr{B} = k$ and, in this
    case, $\Card{\xiD(B)}$ is odd;
  \item if $\Card{D} + \Card{B} = 2k +1$ and $0 \not\in B$, then
    $\Card{D \setminus \xiD(B)} + \MyExpr{B} = k +1$ and, in this
    case, $\Card{\xiD(B)}$ is odd;
  \item if $\Card{D} + \Card{B} = 2k +1 $ and $0 \in B$, then
      $\Card{D \setminus \xiD(B)} + \MyExpr{B} = k$ and, in this case,
      $\Card{\xiD(B)}$ is even.
    \end{enumerate}
  \end{lemma}
  The next lemma restates these properties on the side of preimages.
  \begin{lemma}
    \label{lemma:inverseImages}
    For $C \subseteq \setn$, let $B_{1},B_{2}$ be the two preimages of
    $C$ via $\xiD$. Then, the following statements hold:
    \begin{enumerate}
    \item if $\Card{D \setminus C} + \Myexpr{C} = k$ and $\Card{C}$ is
      even, then the two preimages $B_{1},B_{2}$ of $C$ satisfy
      $\Card{D} + \Card{B_{1}} = 2k$ and
      $\Card{D} + \Card{B_{2}} = 2k +1$;
    \item if $\Card{D \setminus C} + \myexpr{C} = k$ and $\Card{C}$ is
      odd, then the two preimages $B_{1},B_{2}$ of $C$ satisfy
      $\Card{D} + \Card{B_{1}} = 2k -1$ and
      $\Card{D} + \Card{B_{2}} = 2k$.
    \end{enumerate}
  \end{lemma}

\begin{definition}
  For each $n \geq 0$ and $k \in \setnp[2n]$, we let $\WBPnk$ be the
  set of \wbp{s} $(w,B)$ such that $\Card{\Desc{w}} + \Card{B} = k$.
\end{definition}

\begin{proposition}
  \label{prop:ThetaDownK}
  For each $n \geq 0$ and $k \in \setnp$, the restriction of
  $\Theta_{n}$ to $\WBPnk[2k]$ yields a bijection $\Theta_{n,k}$ from
  $\WBPnk[2k]$ to $\SBPnk$.
\end{proposition}
\begin{proof}
    By the first two items of Lemma~\ref{lemma:counting}, the
    restriction of $\Theta_{n}$ to $\WBPnk[2k]$ takes values in
    $\SBPnk$. The restriction map is injective. Indeed, if
    $\Theta_{n}(w,B) = \Theta_{n}(w',B')$, then $w = w'$ and, for
    $D = \Desc{w}$, $C = \xiD(B) = \xiD(B')$.
    Lemma~\ref{lemma:inverseImages} states that each
    $C \subseteq \setn$ has at most one $\xiD$-preimage $B$ satisfying
    $\Card{D} + \Card{B} = 2k$, whence $B = B'$.  This map is also
    surjective: using Lemma~\ref{lemma:inverseImages}, if
    $(w,C)\in \SBP_{n,k}$, $D = \Desc{w}$, and $B_{1} \neq B_{2}$ are
    such that $0 \not \in B_{1}$ and $\xiD(B_{1}) = \xiD(B_{2}) = C$,
    then $(w,B_{1})$ is a preimage of $(w,C)$ if $\Card{C}$ is even,
    and $(w,B_{2})$ is a preimage of $(w,C)$ if $\Card{C}$ is odd.
\end{proof}

Let us recall that, for $u \in \Bn$, $\descBp{u}$ denotes the set of
strictly positive descents of $u$, see
Lemma~\ref{lemma:positiveDescs}.
\begin{definition}
  For each $k \in \setnp[n-1]$, we let $\SBPn^{k}$ be the set of
  \sbp{s} $(w,B) \in \SBPn$ such that $\Card{\descBp{\psi(w,B)}} = k$.
\end{definition}
Let us pinpoint the following characterization of the set $\SBPn^{k}$:
\begin{lemma}
  \label{lemma:SBPnK}
  For each \sbp $(w,C) \in \SBPn$, 
  \begin{align*}
    (w,C) & \in \SBPn^{k}
    \Tiff
    \begin{cases}
      \Card{C} \text{
        is even} \tand     \Card{\desc{w} \setminus C} + \myexpr{C} = k, \tor \\[3mm]
      \Card{C} \text{
        is odd}    \tand \Card{\desc{w} \setminus C} + \myexpr{C} = k +1 \,.
    \end{cases}
  \end{align*}
\end{lemma}
\begin{proof}
  We have 
  \begin{align*}
    \Card{\descBp{\psi(w,C)}} & = k \Tiff
    \begin{cases}
      0 \not\in \descB{\psi(w,C)} \tand \desB{\psi(w,C)} = k, \tor \\[3mm]
      0 \in \descB{\psi(w,C)} \tand \desB{\psi(w,C)} = k+1 \,.
    \end{cases}
  \end{align*}
  The statement of the lemma follows using
  Lemma~\ref{lemma:ZeroDescent} 
  and Proposition~\ref{prop:descentsFromSBPs}.
\end{proof}

\begin{proposition}
  \label{prop:ThetaUpK}
  For each $k \in \setnp[n-1]$, the restriction of  $\Theta_{n}$ to
  $\WBPnk[2k + 1]$
  yields  a bijection $\Theta_{n}^{k}$ from $\WBPnk[2k + 1]$ to $\SBPn^{k}$.
\end{proposition}
\begin{proof}
    By items 3. and 4. in Lemma~\ref{lemma:counting}, and also using
    Lemma~\ref{lemma:SBPnK}, the restriction of $\Theta_{n}$ to
    $\WBPnk[2k +1]$ takes values in $\SBPn^{k}$.  The restriction map
    is injective. Indeed, if $\Theta_{n}(w,B) = \Theta_{n}(w',B')$,
    then $w = w'$ and, for $D = \Desc{w}$, $C = \xiD(B) = \xiD(B')$.
    Lemma~\ref{lemma:inverseImages} states that each
    $C \subseteq \setn$ has at most one preimage $B$ satisfying
    $\Card{D} + \Card{B} = 2k +1$, whence $B = B'$.
    This map is also surjective.  Let $(w,C)\in \SBPn^{k}$,
    $D = \Desc{w}$, and $B_{1} \neq B_{2}$ be such that
    $0 \not \in B_{1}$ and $\xiD(B_{1}) = \xiD(B_{2}) = C$. By
    Lemma~\ref{lemma:inverseImages}, if $\Card{C}$ is even, then
    $(w,C)$ has the preimage $(w,B_{2})$, and if $\Card{C}$ is odd,
    then $(w,C)$ has the preimage $(w,B_{1})$.
\end{proof}

%% file: corollaries.tex
To end this section, we collect the consequences of the bijections
established so far.
\begin{theorem}
  The following relations hold:
  \begin{align*}
    \eulB{n}{k} & = \sum_{i = 0}^{2k}\eulA{n}{i}\binom{n+1}{2k-i}\,,
    \tag*{\eqref{eq:eulBeven}}
    \\
    2^{n}\eulA{n}{k} & = \sum_{i = 0}^{2k+1}\eulA{n}{i}\binom{n+1}{2k
      +1-i}\,.
    \tag*{\eqref{eq:eulBodd}}
  \end{align*}
\end{theorem}
\begin{proof}
  We have seen that \sp{s} $u \in \Bn$ such that $\desB{u} = k$ are in
  bijection (via the mapping $\psi$ of Definition~\ref{def:psi}) with
  \sbp{s} in $\SBPnk$. Next, this set is in bijection (see
  Proposition~\eqref{prop:ThetaDownK}) with the set $\WBPnk[2k]$ of
  \wbp{s} $(w,B) \in \WBPn$ such that $\des{w} + \Card{B} = 2k$. The
  cardinality of $\WBPnk[2k]$ is the right-hand side of equality
  \eqref{eq:eulBeven}.

  The left-hand side of equality \eqref{eq:eulBodd} is the cardinality
  of the set of \sp{s} $u$ such that $\Card{\descBp{u}} = k$, see
  Lemma~\ref{lemma:positiveDescs}. This set is in bijection with the
  set $\SBPn^{k}$ (via $\psi$ defined in~\ref{def:psi} and by the
  definition of $\SBPn^{k}$) which, in turn, is in bijection (see
  Proposition~\eqref{prop:ThetaUpK}) with the set $\WBPnk[2k+1]$ of
  \wbp{s} $(w,B) \in \WBPn$ such that $\des{w}+ \Card{B} = 2k +1$. The
  cardinality of this set is the right-hand side of equality
  \eqref{eq:eulBodd}.
\end{proof}

\begin{theorem}
  The following relation holds:
  \begin{align}
    \label{eq:BsqOne}
    B_{n}(t^{2})   & = (1+t)^{n
      +1} S _{n}(t) - 2^{n}tS_{n}(t^{2})\,.
  \end{align}
\end{theorem}
\begin{proof}
  By \eqref{eq:eulBeven}, $\sEulB{n}{k}$, which is the coefficient of
  $t^{2k}$ in the polynomial $B_{n}(t^{2})$, is also the coefficient
  of $t^{2k}$ in $(1+t)^{n +1}S_{n}(t)$.  By \eqref{eq:eulBodd},
  $2^{n}\sEulA{n}{k}$ is the coefficient of $t^{2k +1}$ in the
  polynomials $2^{n}tS_{n}(t^{2})$ and $(1+t)^{n
    +1}S_{n}(t)$. Therefore
  \begin{align*}
    B_{n}(t^{2}) + 2^{n}tS_{n}(t^{2}) & = (1+t)^{n +1}S_{n}(t)\,,
    \tag*{\eqref{eq:main}}
  \end{align*}
  whence equation~\eqref{eq:BsqOne}.
\end{proof}

%% file: Stembridge.tex
\section{\Stem'{s} identity for \En{s} of type $\D$}
\label{sec:Stembridge}

We recall that a \sp $u \in \Bn$ is \emph{even signed} if the
number of negative letters in its window notation is even. The \esp{s}
of $\Bn$ form a subgroup $\Dn$ of $\Bn$ and in fact the groups $\Dn$
are standard models for the abstract Coxeter groups of type $\D$.

Definitions analogous to those given in Section~\ref{sec:notation} for
the types $\A$ and $\B$ can be given for type $\D$. Namely, for
$u \in \Dn$, we set
\begin{align}
  \label{eq:descD}
  \descD{u} & \eqdef \makebox[1cm][l]{$\displaystyle \set{ i \in \set{0,1,\ldots ,n-1} \mid u_{i} >
      u_{i + 1}}\,,$} 
  \intertext{where we have set $u_{0} = - u_{2}$,
  } 
  \notag 
  \desD{u} & \eqdef \Card{\descD{u}} \,,
  & 
  \eulD{n}{k} & \eqdef \Card{\set{ u \in \Dn \mid \desD{u} = k}}\,,
  &
  D_{n}(t) & \eqdef \sum_{k = 0}^{n}\eulD{n}{k}t^{k}\,.
\end{align}
The formula in \eqref{eq:descD} is the standard one, see e.g. \citep[\S
8.2]{BB2005} or \citep{BagnoBG18}.  The reader will have no
difficulties verifying that, up to renaming $0$ by $-1$, the type $\D$
descent set of $u$ can also be defined as follows, see
\citep[\S 13]{Petersen2015}:
\begin{align}
  \label{eq:descDTwo}
  \descD{u} & \eqdef \set{ i \in \set{-1,1,\ldots ,n-1} \mid u_{i} >
    u_{\abs{i} + 1}}\,,
\end{align}
where now $u_{-1} = \minus{u_{1}}$, as normal if $u$ is written in full notation.

It is convenient to consider a more flexible representation of
elements of $\Dn$. If $u \in \Bn$, then its mate is the \sp
$\mate{u} \in \Bn$ that differs from $u$ only for the sign of the
first letter. Notice that $\mate{\mate{u}} = u$.  We define a
\emph{\fsp} (see \citep[\S 13]{Petersen2015}) as an unordered pair of
the form $\upair{u,\mate{u}}$ for some $u\in \Bn$.  Clearly, just one
of the mates is even signed and therefore \fsp{s} are combinatorial
models of $\Dn$.

The \prep of a \fsp is insensitive of how the diagonal is crossed,
either from the West, or from the North. The following are possible
ways to draw a \fsp on a grid:
\begin{center}
  \includegraphics{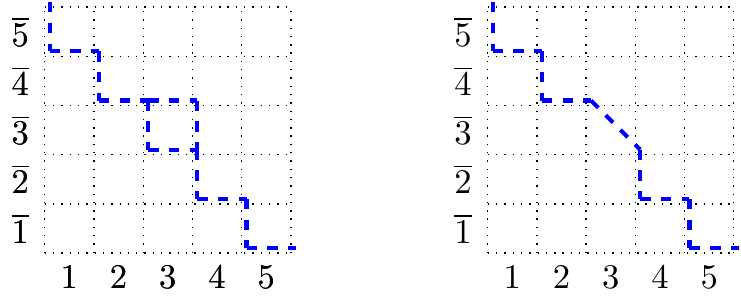}
\end{center}

Even if the formulas in \eqref{eq:descD} and \eqref{eq:descDTwo} have
been defined for \esp{s}, they still can be computed for all
\sp{s}.  
The formula in \eqref{eq:descDTwo} is not invariant under taking
mates, however the following lemma shows that this formula suffices to
compute the number of type $\D$ descents of a \fsp and therefore the
Eulerian numbers $\sEulD{n}{k}$.

\begin{lemma}
  \label{lemma:descmates}
  For each $u \in \Bn$, $1 \in \DescD{u}$ if and only if
  $-1 \in \DescD{\mate{u}}$. Therefore $\desD{u} = \desD{\mate{u}}$.
\end{lemma}
\begin{proof}
  Suppose $1 \in \DescD{u}$, that is $u_{1} > u_{2}$. Then
  $\mate{u}_{-1} = - (-u_{1}) = u_{1} > u_{2}$, and so
  $-1 \in \DescD{u}$. The opposite entailment is proved similarly.

  For the last statement, observe that
  $\DescD{u} = \Delta_{u} \cup \set{i \in \set{2,\ldots ,n-1} \mid
    u_{i} > u_{i +1}}$ with
  $\Delta_{u} \eqdef \set{i \in \set{1,-1} \mid u_{i} > u_{\abs{i}
      +1}}$ and, by what we have just remarked, we have
  $\Card{\Delta_{u}} = \Card{\Delta_{\mate{u}}}$.  It follows that
  $\Card{\descD{u}} = \Card{\descD{\mate{u}}}$.
\end{proof}

Our next aim is to derive \Stem's identity 
\begin{align}
  D_{n}(t) & = B_{n}(t) - n2^{n-1}tS_{n-1}(t) \,, 
  \intertext{see
    \citep[Lemma 9.1]{Stembridge1994}, which, in term of the
    coefficients of these polynomials, amounts to}
  \eulD{n}{k} & = \eulB{n}{k} - n2^{n-1}\eul{n-1}{k-1} \,.
  \label{eq:stembridge}
\end{align}

\begin{definition}
  A \sp $u$ is \emph{\smooth} if $u_{1},u_{2}$ have equal sign and,
  otherwise, it is \emph{\ns}. 
\end{definition}
The reason for naming a \sp \smooth arises again from the \prep of a
\sp: the \smooth \sp is, between the two mates, the one minimizing
the turns nearby the diagonal, as suggested in Figure~\ref{fig:smooth} with two pairs of
mates as examples.
\begin{figure}
  \centering
  \includegraphics{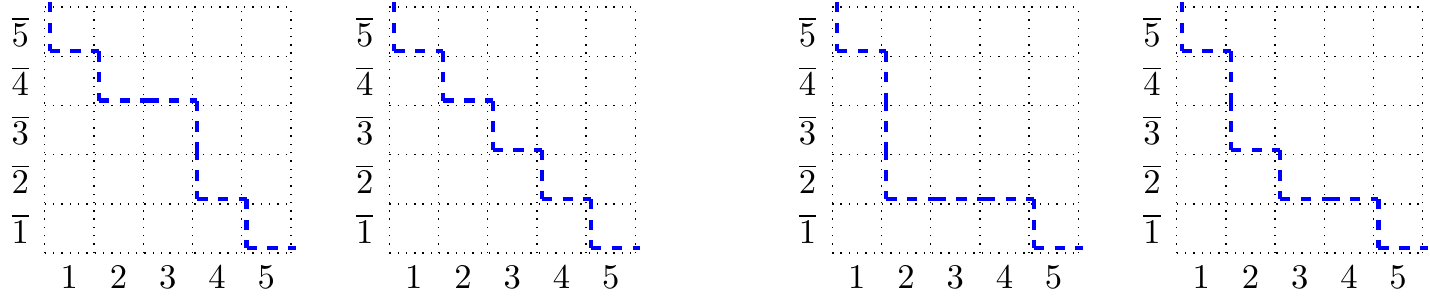}
  \caption{Two pairs of mates, the smooth mates are on the left}
  \label{fig:smooth}
\end{figure}

\begin{lemma}
  \label{lemma:smooth}
  If $u \in \Bn$ is \smooth, then $-1 \in \DescD{u}$ if and only if
  $0 \in \DescB{u}$ and therefore $\desD{u} = \desB{u}$.
\end{lemma}
\begin{proof}
  Suppose $0 \in \DescB{u}$, so $u_{1} < 0$ and $u_{2} < 0$ as well,
  since $u$ is \smooth. Then $u_{-1} = -u_{1} > 0 > u_{2}$, so
  $-1 \in \DescD{u}$.
  Conversely, suppose $-1 \in \DescD{u}$, that is, $u_{-1} > u_{2}$. If
  $u_{1} > 0$, then $0 > - u_{1} = u_{-1} > u_{2}$, so $u_{1},u_{2}$
  have different sign, a \contr. Therefore $u_{1} < 0$ and
  $0 \in \DescB{u}$.
\end{proof}

\begin{corollary}
  \label{cor:smooth}
  There is a bijection between the set of \smooth \sp{s} in $\Bn$
  with $k$ type $\B$ descents and the set of \esp{s} in $\Dn$ with
  $k$ type $\D$ descents.
\end{corollary}
Indeed, if $u \in \Bn$ is \smooth and \es, then we let $v = u$, so
$\descB{u} = \descD{v}$, by Lemma~\ref{lemma:smooth}. If $u$ is
\smooth but not \es, then its mate $\mate{u}$ is \es. We let then
$v = \mate{u}$ and then $\descB{u} = \descD{u} = \descD{\mate{u}}$,
using Lemmas~\ref{lemma:descmates} and \ref{lemma:smooth}.

Next, we consider the correspondence---let us call it $\chi$---sending
a \ns \sp $u \in \Bn$ to the pair $(\abs{u_{1}},u')$, where $u'$ is
obtained from $u_{2}\cdots u_{n}$ by normalising this sequence, so
that it takes absolute values in the set
$\setn[n-1]$. 
For example
$\chi(6\wbar{12347}5) = (6,\wbar{12346}5)$ and
$\chi(\wbar{2}316\wbar{47}5) = (2,215\wbar{36}4)$, as suggested below:
\begin{align*}
  6\wbar{12347}5 & \leadsto (6,\wbar{12347}5)\leadsto
  (6,\wbar{12346}5)\,, &
  \wbar{2}316\wbar{47}5 & \leadsto (2,316\wbar{47}5) \leadsto (2,215\wbar{36}4)\,.
\end{align*}
Notice that this transformation is reversible.  Consider for example
the pair $(3,\bar{1}5\bar{4}23)$.  We can first rename
$\bar{1}5\bar{4}23$ so $3$ is not the absolute value of any letter,
thus obtaining $\bar{1}6\bar{5}24$. We can then add $\pm 3$ in front
of this word, having two choices, $3\bar{1}6\bar{5}24$ and
$\bar{3}\bar{1}6\bar{5}24$. There is exactly one choice yielding a
\ns \sp, namely $3\bar{1}6\bar{5}24$.

The process of normalizing the sequence $u_{2}\ldots u_{n}$ can be
understood as applying to each letter of this sequence the unique
order preserving bijection
$N_{n,x} : \setpmn \setminus \set{x,\wbar{x}} \rto \setpmn[\,n-1]$
where, in general, $x \in \setn$ and, in this case, $x = \abs{u_{1}}$.

\begin{lemma}
  Let $n \geq 2$.  For each pair $(x,v)$ with $x \in \setn$ and
  $v \in \B_{n-1}$, there exists a unique \ns $u \in \Bn$ such that
  $\chi(u) = (x, v)$.
\end{lemma}
\begin{proof}
  We construct $u$ firstly by renaming $v$ to $v'$ so that none of
  $x,\wbar{x}$ appears in $v'$ (that is, we apply to each letter of
  $v$ the inverse of $N_{n,x}$) and then by adding in front of $v'$
  either $x$ or $\wbar{x}$, according to the sign of the first letter
  of $v'$.
\end{proof}
\begin{lemma}
  \label{lemma:nonsmooth}
  The correspondence $\chi$ restricts to a bijection from the set of
  \ns \sp{s} $u \in \Bn$ such that $\desB {u} = k$ to the set of pairs
  $(x,v)$ where $x \in \setn$ and $v \in \B_{n-1}$ is such that
  $\Card{\descBp{v}} = k -1$.
\end{lemma}
\begin{proof}
  We have already argued that $\chi$ is a bijection from the set of
  \ns \sp{s} $u$ of $\setn$ to the set of pairs $(x,v)$ with
  $x \in \setn$ and $v \in \B_{n-1}$. Therefore, we are left to argue
  that, for a \ns $u $ and $v$ such that $\chi(u) = (x,v)$,
  $\desB{u} = k$ if and only if $\Card{\descBp{v}} = k -1$. Said
  otherwise, we need to argue that, for such $u$ and $v$,
  $\Card{\descBp{v}} = \desB{u} -1$.
  To this end, observe that (i) $\Card{\descB{u} \cap \set{0,1}} = 1$,
  since $u_{1},u_{2}$ have different sign, (ii)
  $\descBp{v} = \set{i -1 \mid i \in \descB{u} \cap \set{2,\ldots
      ,n-1} }$, from which the relation
  $\Card{\descBp{v}} = \desB{u} -1$ follows.
\end{proof}

\begin{theorem}
  The 
  following relations hold:
  \begin{align*}
    \eulB{n}{k} & = \eulD{n}{k} + n2^{n-1}\eulA{n-1}{k-1}\,,
    & 
    B_{n}(t) & = D_{n}(t) + n2^{n-1}tS_{n-1}(t) \,.
  \end{align*}
\end{theorem}
\begin{proof}
  Every \sp is either \smooth or \ns.  By Corollary~\ref{cor:smooth},
  the \smooth \sp{s} with $k$ type $\B$ descents are in bijection with
  the \esp{s} with $k$ type $\D$ descents.  By
  Lemma~\ref{lemma:nonsmooth}, the \ns \sp{s} $u \in \B_{n}$ with $k$
  type $\B$ descents are in bijection with the pairs
  $(x,v) \in \setn \times \B_{n-1}$ such $\Card{\descBp{v}} = k
  -1$. Using Lemma~\ref{lemma:positiveDescs}, the number of these
  pairs is $n2^{n-1}\sEulA{n-1}{k-1}$.
\end{proof}

\medskip

%% file: computations.tex
\begin{example}
  We end this section exemplifying the use of formulas
  \eqref{eq:eulBeven} and \eqref{eq:stembridge} by which computation
  of the \En{s} of type $\B$ and $\D$ is reduced to computing \En{s}
  of type $\A$.  Let us mention that our interest in \En{s} originates
  from our lattice-theoretic work on the lattice variety of
  Permutohedra \citep{JEMS} and its possible extensions to generalized
  forms of Permutohedra \citep{Pouzet1995,STA2,CWO}. Among these
  generalizations, we count lattices of finite Coxeter groups in the
  types $\B$ and $\D$ \citep{Bjo84}. While it is known that the
  lattices $\Bn$ span the same lattice variety of the permutohedra,
  see \citep[Exercise 1.23]{STA1}, characterizing the lattice variety
  spanned by the lattices $\Dn$ is an open problem. 
  A first step towards solving this kind of problem is to characterize
  (and count) the \jirr elements of a class of lattices. In our case,
  this amounts to characterizing
  the elements $u$ in $\Bn$ (\resp in $\Dn$) such that $\desB{u} = 1$
  (\resp such that $\desD{u} = 1$).
  The numbers $\sEulB{n}{1}$ and $\sEulD{n}{1}$ are known to be equal
  to $3^{n} - n -1$ and $3^{n} - n -1 - n2^{n-1}$ respectively, see
  \citep[Propositions 13.3 and 13.4]{Petersen2015}.
  Let us see how to derive these identities using the formulas
  \eqref{eq:eulBeven} and \eqref{eq:stembridge}.
  To this end, we also need the alternating sum formula for \En{s},
  see e.g.  \cite[Theorem 1.11]{Bona2012} or \cite[page
  12]{Petersen2015}:
\begin{align}
  \label{eq:alternating} 
  \eulA{n}{k} & = \sum_{j = 0}^{k}(-1)^{j}\binom{n+1}{j}(k
  +1-j)^{n}\,.
\end{align}

For type $\B$, we have
\begin{align*}
  \eulB{n}{1} &  = \eulA{n}{0}\binom{n+1}{2} +
  \eulA{n}{1}\binom{n+1}{1} + \eulA{n}{2}\binom{n+1}{0} \\
  & = \binom{n+1}{2} +
  (2^{n} - n -1)(n +1) + \eulA{n}{2} \\
  & = \binom{n+1}{2} + (2^{n} - n -1)(n +1) + 3^{n} - 2^{n}(n+1) +
  \binom{n + 1}{2} \,,
  \tag*{by \eqref{eq:alternating}}
  \\
  &
  = 3^{n} - (n +1)^{2} + 2\binom{n+1}{2}
  = 3^{n} - (n +1)(n+1 - n) = 3^{n} - n -1  \,.
\end{align*}
The computation of type $\D$ is then immediate from Stembridge's
identity \eqref{eq:stembridge}:
\begin{align*}
  \eulD{n}{1} &  = \eulB{n}{1} - n2^{n-1}\eul{n-1}{0} = 3^{n}  - n - 1
  - n2^{n-1}\,.
  \tag*{\EndOfExample}
\end{align*}
\end{example}

%% file: threshold.tex
\section{\Tg{s} and their \dgo{s}}
\label{threesholds}
Besides presenting the bijective proofs, a goal of this paper is to
illustrate the \prep of \sp{s} and exemplify its potential.  The
attentive reader might object that the \prep is not in use within
Section~\ref{sec:Stembridge}. Indeed, after discovering the bijective
proof of \StemI via the \prep, we realized that the proof could be
simplified and reach a larger audience by avoiding mentioning the
representation. It might be asked then whether the \prep yields more
information, in particular with respect to the lattices of the Coxeter
groups $\Dn$. We answer this question in this section.  The type $\D$
set of inversions of an \esp 
can be defined as follows:
\begin{align*}
  \invD{u} & \eqdef \invB{u} \setminus \set{(-i,i) \mid i \in \setn}\,,
\end{align*}
which, graphically, amounts to ignoring cells on the diagonal:
\begin{center}
  \includegraphics{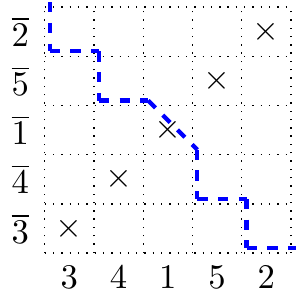}
\end{center}
As mentioned in Remark~\ref{rem:charInversions}, we can identify
the set of inversions of a \sp $u$ with the disjoint union of
$\inv{\lux}$ and a set of unordered pairs. For \esp{s}, this
identification yields:
\begin{align}
  \label{def:Eu}
  \nonumber
  \invD{u} & = \inv{\lux} \cup \Eu\quad \\
  \text{with} \quad
  \Eu & \eqdef \set{\upair{i,j} \mid i,j \in \setn, i \neq j,\,
    ((\lux)^{-1}(i),(\lux)^{-1}(j)) \text{
      lies below } \pi^{u} }\,.
\end{align}
Therefore, we consider $(\setn,\Eu)$ as a simple graph on the set of
\vertices $\setn$.
Let us observe that the definition of the set of edges $E^{u}$
in~\eqref{def:Eu} makes sense for all \sp{s}, not just for an
  \esp{s}. Moreover, for mates $u$ and $\mate{u}$, we have
  $E^{u} = E^{\mate{u}}$.
  Before exploring further the graph $(\setn,\Eu)$, we recall some
  standard graph-theoretic concepts. For an arbitrary simple graph
  $(V,E)$ and a vertex $v \in V$, we let:
 \begin{align*}
   \NE(v) & \eqdef \set{u \in V \mid \upair{v,u} \in E}\,, & \degE(v) & \eqdef
   \Card{\NE(v)}\,,& \SE{v} & \eqdef \NE(v) \cup \set{v}\,.
 \end{align*}
 $\NE(v)$ is the \emph{neighbourhood} of the vertex $v$, $\degE(v)$
 is its \emph{degree}, and $\SE{v}$ is often called the \emph{star} of the
 vertex $v$.
 A linear ordering $v_{1},\ldots ,v_{n}$ of $V$ is a \emph{\dgo} of
 $(V,E)$ if
 $\degE(v_{1}) \geq \degE(v_{2})\geq \ldots \geq \degE(v_{n})$.
 A \emph{preorder} on a set $V$ is a reflexive and transitive binary
 relation on $V$.
 The
 \emph{vicinal preorder} of a graph $(V,E)$, denoted $\vp_{E}$, is
 defined by $v \vp_{E} u \tiff \NE(v) \subseteq \SE{u}$. Notice that
 the relation $\NE(v) \subseteq \SE{u}$ is equivalent to
 $\NE(v) \setminus \set{u} \subseteq \NE(u) \setminus \set{v}$.
 The vicinal preorder is indeed a preorder, see e.g. \citep{MP1995}.
 For completeness, we add a statement and a proof of this fact.
 \begin{lemma}
   The relation $\vp_{E}$ on a simple graph $(V,E)$ is reflexive and
   transitive.
 \end{lemma}
 \begin{proof}
   Reflexivity is
   obvious. For transitivity, let us consider $u,v,w \in V$ such that
   $u \vp_{E} v \vp_{E} w$. If $u,v,w$ are not pairwise distinct, then
   $u \vp_{E} w$ immediately follows.  Therefore, let us assume that
   $u,v,w$ are pairwise distinct with $\NE(u) \subseteq \SE{v}$ and
   $\NE(v) \subseteq \SE{w}$. Let $x \in \NE(u)$. If $x \neq v$, then
   $x \in \NE(v) \subseteq \SE{w}$. If $x = v$, then $v \in \NE(u)$,
   thus $u \in \NE(v) \subseteq \SE{w}$ and since $u \neq w$,
   $u \in \NE(w)$; thus $w \in \NE(u) \subseteq \SE{v}$, so
   $w \in \NE(v)$ and $x = v \in \NE(w)$.
   Therefore $\NE(u) \subseteq \SE{w}$.
 \end{proof}

 Next, we take Theorem 1 in \citep{chvatal1977} as the definition of
 the class of threshold graphs and consider, among the possible
 characterizations of this class, the one that uses the vicinal
 preorder.
 \begin{definition}
   A graph $(V,E)$ is \emph{threshold} if it does not contain an
   induced subgraph isomorphic to one among $2K_{2}$, $P_{3}$ and
   $C_{4}$ (these graphs
   are illustrated in
   Figure~\ref{fig:excludedFromThreshold}).
 \end{definition}
 \begin{figure}
  \centering
  \includegraphics[scale=1]{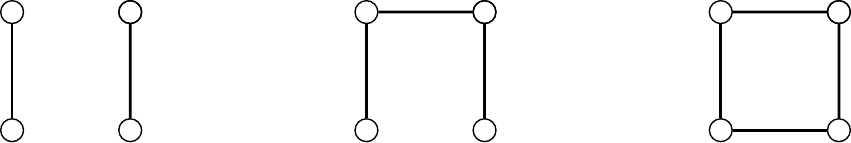}
  \caption{The (unlabelled) graphs $2K_{2}$,
    $P_{3}$, and $C_{4}$}
  \label{fig:excludedFromThreshold}
\end{figure}
A binary relation $R$ on $V$ is \emph{total} if and only if, for each $v,u
\in V$, $v R u$ or $u R v$.

\def\citeMP{\citep[Theorem 1.2.4]{MP1995}}
\begin{proposition}[see e.g. \citeMP]
  \label{prop:tgCharact}
  A graph $(V,E)$ is threshold if and only if the vicinal preorder is
  total.
\end{proposition}
We develop next a few considerations on  \tg{s}.
\begin{lemma}
  \label{lemma:charTgs}
  For a simple graph $(V,E)$ and a total ordering $<$ on $V$, the
  following conditions are equivalent:
  \begin{enumerate}[(i)]
  \item $(V,E)$ is a \tg and $<$ is a \dgo,
  \item $u < v$ implies $v \vp[E] u$, for each $v,u \in V$. 
  \end{enumerate}
  If any of the above conditions hold, then, for each $u \in V$,
  $\NE(u)$ is a downset in the following sense: if $v \in \NE(u)$ and
  $w \neq u$ is such that $w < v$, then $w \in \NE(u)$.
\end{lemma}
\begin{proof}
  We observe firstly that $v \vpE u$ implies $\degE(v) \leq
  \degE(u)$. Indeed, this follows from the fact that $v \vpE u$
  amounts to
  $\NE(v) \setminus \set{u} \subseteq \NE(u) \setminus \set{v}$ and
  that $u \in \NE(v)$ if and only if $v \in \NE(u)$. Notice also that
  the same argument can be used to argue that if $v \vpE u$ and
  $ u \not\vpE v$, then $\degE(v) < \degE(u)$.
  
  Let therefore $(V,E)$ and $<$ be as stated. By the remark above, if
  $<$ satisfies (ii) then it is a \dgo and the relation $\vpE$ is
  total, since if $u \not \vpE v$, then $u \not< v$, so $v \leq u$ and
  $v \vpE u$; thus $(V,E)$ is a \tg.
  Suppose next $(V,E)$ is a \tg and that $<$ is a \dgo, so $u < v$
  implies $\degE(v) \leq \degE(u)$. Let $u < v$ and suppose that
  $v \not\vpE u$. Since the vicinal preorder is total, we have then
  $u \vpE v$ and so $\degE(u) < \degE(v)$, contradicting
  $\degE(v) \leq \degE(u)$.

  For the last statement, for such $u,v,w$, the relation $w < v$
  implies
  $\NE(v) \setminus \set{w} \subseteq \NE(w) \setminus \set{v}$.
  Since $v \in \NE(u)$, then $u \in \NE(v) \setminus \set{w}$ and
  $u \in \NE(w) \setminus \set{v}$, so $w \in \NE(u)$.
\end{proof}

We establish now the connection between \tsg{s} on the set of \vertices
$\setn$, paths, and \esp. We achieve this through the order-theoretic
notion of \Gc, see e.g. \citep[\S 7]{DP2002} or \citep{Nelson1976}.  A
\emph{Galois connection} on $\setnp$ is a pair of functions
$f,g : \setnp \rto \setnp$ such that, for each $x,y \in \setnp$,
$y \leq f(x)$ if and only if $x \leq g(y)$.  We say that a map
$f : \setnp \rto \setnp$ is a \emph{\hf} (we shall discover few lines
below the reason for the naming) if it is is antitone---that is,
$x \leq y$ implies $f(y) \leq f(x)$, for each $x,y \in \setnp$---and
moreover $f(0) = n$.  Observe that, for a \Gc $(f,g)$, $f$ is a
\hf. It is part of elementary order theory that a map
$f : \setnp \rto \setnp$ is part of a Galois connection exactly when
it is a \hf. Moreover, for $f: \setnp \rto \setnp$, there is at most
one function $g : \setnp \rto \setnp$ such that $(f,g)$ is a Galois
connection.

Paths from $(0,n)$ to $(n,0)$ that are composed only by East and
\Ss{s} bijectively correspond to \hf{s}. The bijection is realized by
the correspondence sending a path $\pi$ to
$\height_{\pi} : \setnp \rto \setnp$ such that $\height_{\pi}(x)$ is
the height of $\pi$ after $x$ \Es{s}; Figure~\ref{fig:pathsGCs}
illustrates this correspondence.
\begin{figure}
  \centering
  \includegraphics[page=4]{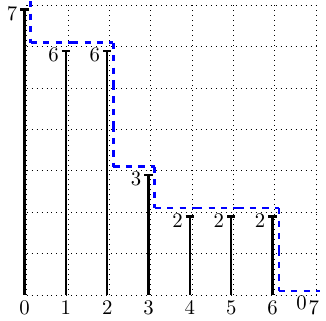}
  \caption{Paths as \hf{s}}
  \label{fig:pathsGCs}
\end{figure}
We refer the reader to \citep{2019-WORDS} for the correspondence
between paths and this kind of functions in the discrete setting. 
For a \hf $f : \setnp \rto \setnp$, we say that $x$ is
  \emph{negative} if $x \leq f(x)$ and that $x$ is \emph{positive} if
  $f(x) \leq x$.  Notice that $x$ is both positive and negative if and
  only if it is \fp of $f$.  Let $N_{f}$ (resp., $P_{f}$) be
  the set of negative (resp., positive) elements of $f$. Observe that
  $N_{f} \neq\emptyset$, so we can define $\gamma_{f}$, the center of
  $f$, as the maximum of this set, $\gamma_{f} \eqdef \max N_{f}$.
  \begin{lemma}
    \label{lemma:fixepointgammaf}
    For a \hf
    $f : \setnp \rto \setnp$, $N_{f}$ is a downset, $P_{f}$
    is an upset, and $\Card{N_{f} \cap P_{f}} \leq 1$. In particular,
    $f$ has at most one \fp, necessarily $\gamma_{f}$.
  \end{lemma}
  \begin{proof}
    If $x \leq y \leq f(y)$, then $f(y) \leq f(x)$, so $N_{f}$ is a
    downset. Similarly, $P_{f}$ s an upset. Also, if $x$ is positive,
    then $f(x)$ is negative, and if $x$ is negative, then $f(x)$ is
    positive.
    Next, the intersection $N_{f} \cap P_{f}$ can have at most one
    element. Indeed, if $x,y$ are distinct \fp{s} and $x < y$,
    then $y = f(y) \leq f(x) = x$, a contradiction.
  \end{proof}

  \begin{lemma}
    For $f = \height_{\pi}$, $(\gamma_{f},\gamma_{f})$ is the
    (necessarily unique) intersection point of $\pi$ with the
    diagonal.
  \end{lemma}
  \begin{proof}
    A straightforward geometric argument shows that the intersection
    point of $\pi$ with the diagonal exists and is unique.  Let
    $x = \gamma_{f}$, so $x \leq \height_{\pi}(x)$ and
    $\height_{\pi}(x +1) \leq x$. Thus, at time $x$, $\pi$ moves from
    $(x,\height_{\pi}(x))$ to $(x,\height_{\pi}(x+1))$ through a
    sequence of South steps. Since
    $\height_{\pi}(x+1) \leq x \leq \height_{\pi}(x)$, the path $\pi$
    meets the point $(x,x)$.
  \end{proof}  

Let us say now that a \hf $f$ is \emph{\sa} if $(f,f)$ is a \Gc. That
is, $f$ is \sa if $y \leq f(x)$ is equivalent to $x \leq f(y)$, for
each $x,y \in \setnp$. We say $f$ is \emph{\fpf} if $f(x) \neq x$, for
each $x \in \setnp$.

\begin{lemma}
  \label{lemma:EastStep}
  For a \hf $f: \setnp \rto \setnp$, let $\pi$ be the unique path
  such that $\height_{\pi} = f$. Then $f$ is \sa if and only if
  $\pi$ is symmetric along the diagonal, in which case $f$ is \fpf
  if and only if its first step after meeting the diagonal is an
  East step.
\end{lemma}
\begin{proof}
  For a path $\pi$, let $\pi'$ be the path obtained from $\pi$ by
  reflecting it along the diagonal.
    It is straightforward that $y \leq \height_{\pi}(x)$ if and only
    if the point $(x,y)$ lies below and on the left of $\pi$, if and
    only if the point $(y,x)$ lies below and on the left of $\pi'$, if
    and only if $x \leq \height_{\pi'}(y)$. Therefore, by identifying
    paths with \hf{s}, the adjoint of $\pi$ is the path obtained by
    reflecting along the diagonal.
    In particular, $\pi$ is \sa if and only if $\pi$ is symmetric
    along the diagonal.

    Let us argue that, for $f = \height_{\pi}$ \sa, $f$ is \fpf if and
    only if $\pi$'s first step after meeting the diagonal is towards
    East.  
    Let $(\gamma_{f},\gamma_{f})$ be the intersection point of $\pi$
    with the diagonal of $\setnp$.  We shall verify whether
    $\gamma_{f}$ is a \fp, since, by
    Lemma~\ref{lemma:fixepointgammaf}, it is the only candidate with
    this property.  If the following step is a South step, then the
    last step before meeting the diagonal is an East step, which
    implies that $\gamma_{f} = \height_{\pi}(\gamma_{f})$, thus $f$
    has a \fp.
    If the following step is an East step, then the last step before
    meeting the diagonal is a South step, which implies that
    $\gamma_{f} < \height_{\pi}(\gamma_{f}) = f(\gamma_{f})$, so $f$
    is \fpf.
  \end{proof}

\begin{proposition}
  \label{prop:dgtgGalois}
  For $f : \setnp \rto \setnp$ a \sa \hf, define
  \begin{align*}
    E_{f} & \eqdef \set{\upair{x,y} \mid x,y \in \setn,\, x\neq y,\, y \leq f(x)}\,.
  \end{align*}
  Then $(\setn,E_{f})$ is a \tg and $<$ is a \dgo of $(\setn,E_{f})$.
  The mapping $f \mapsto E_{f}$ restricts to a bijection from the set
  of \sa \hf{s} on $\setnp$ to the set of \tg{s} of the form
  $(\setn,E)$ such that the standard linear ordering of $\setn$ is a
  \dgo.
\end{proposition}
\begin{proof}
  If $y < x$ and $z \leq f(x)$, then $z \leq f(x) \leq f(y)$, since
  $f$ is antitone. As a consequence, if $y < x$, then
  $\NE[E_{f}](x) \subseteq \NE[E_{f}](y) \cup \set{y} =
  \SE[E_{f}]{y}$, so $(\setn,E_{f})$ is a \tg and $<$ is a \dgo, by
  Lemma~\ref{lemma:charTgs}.

  Conversely, let $([n],E)$ be a \tg for which the standard 
  ordering is a \dgo. As we have seen, $\NE(x)$ is a
  downset: if $y \in \NE(x)$ and $z \neq x$ is such that $z < y$, then
  $z \in \NE(x)$.  
  Define then $f_{E}(x) \eqdef \max \NE(x)$,
  with the conventions that $\max \emptyset = 0$ and
  $\NE(0) = \setnp$, so $f_{E} : \setnp \rto \setnp$.  Observe that
  the following equivalences holds, by the definition of $f_{E}$ and
  the fact that $\NE(x)$ is a downset: $\upair{x,y} \in E$ if and only
  if $y \in \NE(x)$ if and only if $x \neq y$ and $y \leq f_{E}(x)$.
  It immediately follows that $y \leq f_{E}(x)$ if and only if
  $x \leq f_{E}(y)$, so $f_{E}$ is \sa; $f_{E}$ is \fpf since
  $x \not\in \NE(x)$.

  It is easily seen that $E_{f_{E}} = E$ and, whenever $f$ is \fpf,
  $f_{E_{f}} = f$ so under the latter hypothesis the two
  transformations are inverse to each other.
\end{proof}

Let in the following $\TGn$ be the set of pairs $(w,E)$ such that
$(\setn,E)$ is a \tg and $w \in \Sn$ is a \dgo of $(\setn,E)$.  We can
state now the main result of this section:
\begin{theorem}
  \label{thm:tgdo}
  The mapping sending $u$ to $(\lux,\Eu)$ restricts to a bijection
  from  $\Dn$ to $\TGn$.
\end{theorem}
\begin{proof}
  Firstly, we claim that the pair $(\lux,E^{u})$ is constructed
  through intermediate steps, as suggested in the following diagram
  (the notation being used is explained immediately after):
  $$
  \begin{tikzcd}
    u \in \Dn \arrow[mapsto]{r}{} \arrow[mapsto]{dd}{} &
    (\lux,\pi^{u})  \arrow[mapsto]{r}{} &
    (\lux,\pi^{\lowermate{u}})\arrow[mapsto]{d}{} \in \Sn \times \Pin\\
    &&  (\lux,\height_{\pi^{\lowermate{u}}}) \arrow[mapsto]{d}{} \in \Sn \times \HFn\\
    (\lux,E^{u}) \arrow[equal]{rr}{} && (\lux,\lux \circ
    E_{\height_{\pi^{\lowermate{u}}}}) \in \TGn
  \end{tikzcd}
  $$
  We also claim that each step in the upper leg of the diagram yields
  a bijection.  In a second time, we shall verify that the pairs that
  may appear in the bottom right corner are exactly the elements of
  $\TGn$.
  
  We explain the notation used in the diagram.
  For $u \in \Dn$, we let $\lowermate{u} \in \set{u,\mate{u}}$ be such
  that $\lowermate{u}_{1} > 0$. The first step of
  $\pi^{\lowermate{u}}$ after crossing the diagonal is an East step
  and therefore the \hf corresponding to $\pi^{\lowermate{u}}$ is
  \fpf.  We let $\Pin$ denote the set of East and \Ss paths from
  $(0,n)$ to $(n,0)$ that make an East step after meeting the
  diagonal, and that are symmetric along the diagonal. We let $\HFn$
  denote the set of \fpfsahf{s} of $\setnp$. Then the $\height$
  function is a bijection from $\Pin$ to $\HFn$. For $f \in \HFn$,
  $E_{f}$ is the set of edges defined in the statement of
  Lemma~\ref{prop:dgtgGalois}. Finally, if $E$ is a set of edges on
  the \vertices $\setn$ and $\sigma \in \Sn$, then we let
  \begin{align*}
    \sigma \circ E & \eqdef
    \set{\upair{\sigma(x),\sigma(y)} \mid \upair{x,y}\in E} =
    \set{\upair{i,j}\mid \upair{\sigma^{-1}(i),\sigma^{-1}(j)} \in E}\,.
  \end{align*}
  We
  justify now the equality on the bottom line of the diagram.  Notice that
  \begin{align*}
    E^{u} & = \lux \circ E_{\pi^{u}} \qquad\text{with}\qquad
    E_{\pi^{u}} \eqdef \set{\upair{x,y}
      \mid x,y \in \setn, x \neq y, (x,y) \text{ lies below } \pi^{u}}
  \end{align*}
  and that $E_{\pi^{u}} = E_{\pi^{\lowermate{u}}} = E_{f}$ where
  $f = \height_{\pi^{\lowermate{u}}}$.  Indeed, the condition that
  $(x,y)$ lies below $\pi^{\lowermate{u}}$ amounts to saying that $y$
  is less than the height of $\pi^{\lowermate{u}}$ after $x$ \Es{s}.
  This shows that $E^{u} = \lux \circ E_{f}$ with
  $f = \height_{\pi^{\lowermate{u}}}$.
  
  Finally, the following equivalences are clear: $f$ is a \fpfsahf on
  $\setnp$ if and only if $<$ (the ordering given by the identity
  permutation) is a \dgo of the \tg $(\setn,E_{f})$ (by
  Proposition~\ref{prop:dgtgGalois}), if and only if the ordering
  given by the permutation $\sigma$ is a \dgo for the \tg
  $\sigma \circ E$.
  Thus, in the right bottom corner of the above diagram we have all
  the pairs $(w,E)$ such that $(\setn,E)$ is a \tg and the linear
  ordering given by the permutation $w \in \Sn$ is among its \dgo{s}.
\end{proof}

\begin{remark}
  \label{rem:smoothpick}
  Let $f $ be a choice of mates, that is, a function
  $f : \Bn \rto \Bn$ such that
  $f(u) = f(\mate{u}) \in \set{u,\mate{u}}$. Let
  $\Dn^{f} = \set{f(u) \mid u \in \Bn}$. In view of
  $E^{u} = E^{\mate{u}}$ and $\lux = \lux[\mate{u}]$, the same
  argument appearing in the proof of Theorem~\ref{thm:tgdo} shows that
  the mapping $u \mapsto (\lux,E^{u})$ restricts to a bijection from
  $\Dn^{f}$ to $\TGn$.
  In particular, we shall consider the function $sm : \Bn \rto \Bn$
  picking the unique smooth mate $sm(u) \in \set{u,\mate{u}}$. Then
  $\Dn^{sm} $ is in bijection with $\TGn$.
\end{remark}

Theorem~\ref{thm:tgdo} yields a natural representation of the weak
ordering on $\Dn$ as follows.   Order $\TGn$ by saying that
$(w_{1},E_{1}) \leq (w_{2},E_{2})$ if and only if $w_{1} \leq w_{2}$
in the weak ordering of $\Sn$ and, moreover, $E_{1} \subseteq E_{2}$,
so $\TGn$ is clearly a poset.
\begin{theorem}
  The poset $\TGn$ is a lattice isomorphic to the weak ordering of the
  Coxeter group $\Dn$.
\end{theorem}
Notice that $\TGn$ is only loosely related to the lattice of \tg{s} of
\citep{MerrisRoby2005} where unlabeled (that is, up to isomorphism)
\tg{s} are considered. While many are the remarks that we already
could develop using this characterisation of the weak order on $\Dn$,
it is in the scope of future research to complete them and to give a
satisfying description of this ordering.

%% file: countingTGs.tex
\section{Counting \tg{s}}
  Proposition~\ref{prop:dgtgGalois} and Theorem~\ref{thm:tgdo},
  originally conceived for achieving a better understating of the
  structure of the lattices of Coxeter groups of type \D, can also be
  used to enumerate \tg{s}.  As the number of paths from $(0,n)$ to
  $(n,0)$ that are symmetric along the diagonal and begin with an East
  step is easily seen to be $2^{n-1}$,
  Proposition~\ref{prop:dgtgGalois} yields a simple proof that the
  number of unlabeled \tg{s} is $2^{n-1}$.
  Enumeration results for labeled \tg{s} appear in
  \citep{BP1987} and, more recently, in \citep{Spiro2020,GWZ2022} (see
  also \citep[Exercise 5.25]{StanleyHypArr} and \cite[Exercise
  3.115]{Stanley1}).
  Theorem~\ref{thm:tgdo} can be used to give bijective proofs of
  these results. The ideas that we expose next have been suggested by
  the bijection described in \cite[\S 1]{Spiro2020}, that we adapt
  here to fit the correspondences between \tg{s} coming with a \dgo,
  signed and \esp{s}, and \sbp{s}.  Our starting point is the
  following observation:
  \begin{lemma}
    \label{lemma:canonicaldgo}
    Let $(\setn,E)$ be a threshold graph. Then, there exists a unique
    degree ordering $w$ of $(\setn,E)$ such that, if $\degE(i) = \degE(j)$
    and $i < j$, then $w^{-1}(i) < w^{-1}(j)$.
  \end{lemma}
  Recall that $w^{-1}(i) < w^{-1}(j)$ means that $i$ occurs before $j$
  in the permutation $w$, written as a word.  The proof of the Lemma,
  straightforward, amounts to the remark that, given a \dgo, we can
  permute \vertices of equal degree and, in doing so, obtain a \dgo.
  We call the ordering of Lemma~\ref{lemma:canonicaldgo} the
  \emph{\cdgo} of $(\setn,E)$.

  In order to count \tg{s} we can count \sbp{s} $(w,B)$ that, along
  the ideas developed in the previous section, bijectively correspond
  to some $(w,E)$ such that $w$ is the \cdgo of $(\setn,E)$.  To this
  goal, recall that, for $(w,B)$ a \sbp, the bars $B$ split $\setn$
  into blocks, the last of which might be empty.  Observe that
  $i,j \in \setn$ appear in the same block of $(w,B)$ if and only if
  they have equal height, by which we mean that, with $u = \psi(w,B)$
  (cf. Definition~\ref{def:psi}),
  $\height_{\pi^{u}}(w^{-1}(i)) = \height_{\pi^{u}}(w^{-1}(j))$.
  \begin{definition}
    We say that a \sbp $(w,B)$ is \emph{\normal} if, whenever $i,j$ belong
    to the same block and $i < j$, then $w^{-1}(i) < w^{-1}(j)$.
  \end{definition}
  Observe from now that an ordered partition of $\setn$ yields two
  \normal \sbp{s}.  The partition is written as a word, where the
  blocks, already ordered, are written in increasing order and
  separated by a bar. This is the standard construction allowing to
  compute the number of partitions of $\setn$ into $k$ blocks from the
  Eulerian numbers, see e.g. \citep{GS78} or \citep[Theorem
  1.17]{Bona2012}. If we add to the same word a bar in the last
  position $n$, then we obtain a second \sbp. The following
  definition, more involved, shall receive a more intuitive meaning
  with the lemma that immediately follows.
  \begin{definition}
    The \emph{central block} of a \sbp $(w,B)$ is the $k$-th block,
    with $k = \ceiling{\frac{\Card{B} +1}{2}}$.
  \end{definition}
  \begin{lemma}
    Let $(w,B)$ be a \sbp, let $u = \psi(w,B)$, and consider the path
    $\pi^{u}$. If $\pi^{u}$ makes an East-South turn when meeting the
    diagonal, then the central block is the block of equal height
    immediately before the diagonal.  If $\pi^{u}$ makes a South-East
    turn when meeting the diagonal, then the central block is the
    block of equal height immediately after the diagonal.
  \end{lemma}
  \begin{proof}
    Let, as before, $f = \height_{\pi^{u}}$ and $\gamma_{f}$ be such
    that $\pi_{u}$ meets the diagonal in $(\gamma_{f},\gamma_{f})$.
    
    Suppose that $\pi^{u}$ makes an East-South turn when meeting the
    diagonal. Since bars bijectively correspond to East-South turns
    (cf. Lemma~\ref{lemma:ZeroDescent}) and $\pi^{u}$ is symmetric
    along the diagonal, then $\Card{B}$ is odd, say $\Card{B} =
    2k-1$. The path $\pi^{u}$ therefore makes $k-1$ East-South turns
    strictly on the left of the diagonal.  Thus, there are $k-1$
    vertical bars strictly on the left of $\gamma_{f}$ and therefore,
    for $k = \ceiling{\frac{\Card{B} + 1}{2}}$, the $k$-th block is
    the group of equal height immediately before meeting the
    diagonal---that is, the block containing $w(\gamma_{f})$.

    If $\pi^{u}$ makes a South-East turn when meeting the diagonal,
    then $\Card{B}$ is even. Say $\Card{B} = 2k-2$, so
    $k = \ceiling{\frac{\Card{B} + 1}{2}}$.
    The path $\pi^{u}$ makes $k-1$ East-South turns before meeting the
    diagonal, and therefore there are $k-1$ bars on the left of
    $\gamma_{f}$. Since $\pi^{u}$ makes a South-East turn when
    meeting the diagonal, the last of these bars is in position $\gamma_{f}$.
    The $k$-th block immediately occurs after $\gamma_{f}$: it is the
    block of equal height containing $w(\gamma_{f} +1)$.
  \end{proof}

  \begin{definition}
    We say that a \sbp $(w,B)$ is \emph{compatible} if, for
    $u = \psi(w,B)$ and $i,j \in \setn$,
    $\degE[E^{u}](i) = \degE[E^{u}](j)$ if and only if $i,j$ appear in
    the same block of $B$.
  \end{definition}

  \begin{lemma}
    \label{lemma:smoothcentral}
    For a \sbp $(w,B)$, the following are equivalent:
    \begin{enumerate}\item 
      $(w,B)$ is compatible,
    \item  $\psi(w,B)$ is
      smooth,
    \item the central block has at least two elements.
    \end{enumerate}
  \end{lemma}
  \begin{proof}
    Let us argue that 1. is equivalent to 2.
    For $u = \psi(w,B)$, consider the path representation
    $(\pi^{u},\lux,\luy)$ and the \tg $(\setn,E^{u})$, so $w =
    \lux$. For readability, we also let $f = \height_{\pi^{u}}$.
    Recall that $i,j$ appear in the same block of $B$ if and only if
    $f(w^{-1}(i)) = f(w^{-1}(j))$.  On the other hand,
    $\degE[E^{u}](i) = f(w^{-1}(i)) -1$, if
    $w^{-1}(i) \leq f(w^{-1}(i))$ and, otherwise,
    $\degE[E^{u}](i) = f(i)$.  That is, up to a renaming of \vertices,
    the degree is computed as the height modulo a normalisation by one
    before meeting the diagonal.
    Therefore $(w,B)$ is not compatible, if and only if, for some
    $x,y$ such that $x \leq f(x)$ and $f(y) < y$, we have
    $f(x) -1 = f(y)$. Considering that $x < y$ and that $f$ is
    antitone, this happens exactly when
    $f(\gamma_{f}) -1 = f(\gamma_{f} + 1)$.  The latter condition
    holds exactly when $\pi^{u}$, immediately after meeting the
    diagonal, either takes a South step followed by an East step, or
    takes an East step followed by a South step. In turn, this
    condition amounts to saying that $u$ is not smooth.

    Let us argue that 2. is equivalent to 3.
    If $\pi^{u}$ makes an East-South turn when meeting the diagonal,
    then $u$ is smooth if and only if also its second step after
    meeting the diagonal is a South step, if and only if, before
    meeting the diagonal, $\pi^{u}$ takes two East steps, if and only
    if the central block has at least two elements.
    If $\pi^{u}$ makes an South-East turn when meeting the diagonal,
    then $u$ is smooth if and only if also its second step after
    meeting the diagonal is an East step, if and only if the central
    block has at least two elements.
  \end{proof}

  \begin{theorem}
    \label{thm:bijtgsbps}
    There is a bijection between the set of \tg{s} on the vertex set
    $\setn$ and the \normal \sbp{s} whose central block has at least
    two elements. The bijection sends \vertices of equal degree to
    \vertices in the same block.
  \end{theorem}
  \begin{proof}
    Given the \tg $(\setn,E)$, let $w$ be its \cdgo. 

    For $(w,E) \in \TGn$, we let $u \in \Bn$ be the unique \smooth \sp
    corresponding to $(w,E)$ under the bijection described in
    Theorem~\ref{thm:tgdo} and Remark~\ref{rem:smoothpick}, and let
    $(w',B) = \psi^{-1}(u)$. Notice that $w = w'$.
    Then, by Lemma~\ref{lemma:smoothcentral}, the central block of
    $(w,B)$ has cardinality at least two, since it corresponds to a
    smooth \sp.
    Since two \vertices have
    equal degree if and only if they belong to the same block,
    $(w,B)$ is \normal if and only if $w$ is the \cdgo of $(\setn,E)$.
  \end{proof}

  The following corollary allows us to simplify the counting arguments
  that follow.
  \begin{corollary}
    There is a bijection between the set of \tg{s} on the set $\setn$
    and the set of \normal \sbp{s} such that the first block has at
    least two elements.
  \end{corollary}
  \begin{proof}
    Given a \normal \sbp whose central block has at least two
    elements, we move this central block in first position by
    permuting the blocks. This transformation is reversible, since we
    can determine the position where to move back the first block from
    the cardinality of $B$.
  \end{proof}
  
  We can count then the labeled \tg{s} on $\setn$ according to the
  number $i$ of blocks of equal
  degree. 
  By the previous considerations, this amounts to counting \normal
  \sbp{s} whose first block has at least two elements. Considering
  that \normal \sbp{s} are sort of ordered partitions, recall that
  $\smallStirling{n}{i}$, the Stirling number of the second kind, counts
  the number of unordered partitions of an $n$-element set into $i$
  blocks.  We immediately recover the formula from \citep[\S 3]{BP1987}
  for the number $T_{n,i}$ of \tg{s} on $n$ \vertices with $i$
  different degrees:
  \begin{align*}
    T_{n,i} & = 2 \cdot (\fact{i}\cdot\stirling{n}{i} -
    n\cdot\fact{(i-1)}\stirling{n-1}{i-1}) = 2 \cdot
    \fact{(i-1)}\cdot(i\cdot\stirling{n}{i} -
    n\cdot\stirling{n-1}{i-1})\,.
  \end{align*}
  The formula can be understood as follows: out of all the ordered
  partitions of $\setn$ into $i$ blocks, eliminate those whose first
  block is a singleton; transform then such an ordered partition into
  a \sbp by adding or not a bar in the last position. Notice that,
  under the bijection, there is a bar in the last position if and only
  if no vertex is isolated (i.e. has degree $0$).

  It is well known that an ordered partition, when transformed into a
  \sbp as before, is determined by its set of descent positions (that
  are always barred), and by the set of the other barred positions
  (ascent positions, those that are not descent positions).  Whence,
  we can count \sbp{s} according to the number of descents.  This
  yields the formula from \citep{Spiro2020} for counting the number
  $T_{n}$ of \tg{s} as the sum of the numbers $\tau_{n,k}$:
  \begin{align*}
    T_{n} & = \sum_{k = 0,\ldots n-2} \tau_{n,k}\,,
    &&\quad\text{with}\quad & \tau_{n,k} & = 2 \cdot P(n,k)\cdot
    2^{n-2-k} = P(n,k)\cdot
    2^{n-1-k}\,.
  \end{align*}
  Above $P(n,k) = (k +1)\cdot \smallEul{n-1}{k}$ is the number of
  permutations of the set $\setn$ having exactly $k$ descents, and
  whose first block has at least two elements, see \cite[Lemma
  6]{Spiro2020}.
  The number $\tau_{n,k}$ can be interpreted as the number of \tg{s}
  whose ordered partition determined by equal degree has exactly $k$
  descents. The explicit formula for $\tau_{n,k}$ stems from the fact
  that, in order to construct such a partition, we can choose a
  permutation whose first block has at least two elements, and then
  add other bars at ascent positions except for the first ascent
  position.

  That \tg{s} are related to the families $\B$ and $\D$ in the theory
  of Coxeter groups has already been observed, see
  e.g. \citep{EdelmanReiner1994}, \citep[Exercise
  5.25]{StanleyHypArr}, and \citep[Exercise 3.115]{Stanley1}.  It
  needs to be emphasized, however, that the way we came up with \tg{s}
  is orthogonal to the way \tg{s} are being used in these works.  As
  part of future research, we wish to investigate the bijections
  presented in Theorems~\ref{thm:tgdo} and~\ref{thm:bijtgsbps} (which
  can be adapted to fit the type $\B$) with the goal of understanding
  whether they play any role with respect to the problem, dealt with
  in \citep{EdelmanReiner1994}, of characterizing free
  sub-arrangements of the Coxeter arrangements of type $\B$. We also
  aim to understand whether the connection with \tg{s} established
  here can shed some light on the problem of giving combinatorial
  characterisations of the notion of free arrangement \citep[\S
  4]{StanleyHypArr}.

%% file: ack.tex
\acknowledgements
The author is thankful to the
referees for their help to improve a first version of this paper and
for their numerous hints.